\documentclass[%
    a4paper,
    12pt,
    DIV=12,
    BCOR=0mm,
    twoside=semi,
    headings=small
]{scrartcl}

\usepackage{scrpage2}
\usepackage{bm}
\usepackage{nicefrac}				
\usepackage{environ}
\usepackage{amsmath}
\usepackage{xcolor}

\usepackage{calligra}
\usepackage{mathrsfs}

\addtokomafont{title}{\normalfont \bfseries}
\addtokomafont{section}{\normalfont \bfseries}
\addtokomafont{subsection}{\normalfont \bfseries}
\addtokomafont{subsubsection}{\normalfont \bfseries}
\addtokomafont{paragraph}{\normalfont \bfseries}
\newcommand{\quoteparagraph}[1]{\noindent{\normalsize \bfseries #1}\hspace{1em}}

\usepackage[utf8]{inputenc}           
\usepackage{amsmath, amsthm, amssymb} 
\usepackage[english]{babel}           
\usepackage{bibgerm}                  
\usepackage{graphicx}                 
\usepackage{hyperref}                 
\usepackage{booktabs}                 
\usepackage{color}

\newcommand{\R}{\ensuremath{\mathbb{R}}}

\newcommand{\N}{\ensuremath{\mathbb{N}}}
\newcommand{\E}{\ensuremath{\mathbb{E}}}
\renewcommand{\P}{\ensuremath{\mathbb{P}}}

\newcommand{\cA}{\ensuremath{\mathcal{A}}}
\newcommand{\cD}{\ensuremath{\mathcal{D}}}
\newcommand{\cK}{\ensuremath{\mathcal{K}}}
\newcommand{\cQ}{\ensuremath{\mathcal{Q}}}

\newcommand{\cU}{\ensuremath{\mathcal{U}}}
\newcommand{\cZ}{\ensuremath{\mathcal{Z}}}

\newcommand{\Se}{\ensuremath{\mathbb{S}}}
\newcommand{\Pe}{\ensuremath{\mathbb{P}}}

\newcommand{\cDD}{\ensuremath{\cD_{\textnormal{cur}}}}
\newcommand{\rk}[1]{\ensuremath{\left(#1\right)}}
\newcommand{\ek}[1]{\ensuremath{\left[#1\right]}}
\newcommand{\gk}[1]{\ensuremath{\left\{#1\right\}}}

\newcommand{\Shat}{\widehat{S}}
\newcommand{\xhat}{\widehat{x}}

\newcommand{\Lhat}{\widehat{\ell}}

\renewcommand{\rho}{{\mathfrak{r}}}
\renewcommand{\log}{{\ln}}

\DeclareMathOperator{\TWR}{TWR}
\DeclareMathOperator{\HPR}{HPR}
\DeclareMathOperator{\linEQ}{linEQ}

\newtheorem{theorem}{Theorem}[section]
\newtheorem{definition}[theorem]{Definition}
\newtheorem{assumption}[theorem]{Assumption}
\newtheorem{lemma}[theorem]{Lemma}
\newtheorem{remark}[theorem]{Remark}
\newtheorem{corollary}[theorem]{Corollary}
\newtheorem{example}[theorem]{Example}
\newtheorem{examples}[theorem]{Examples}
\newtheorem{setup}[theorem]{Setup}

\numberwithin{equation}{section}

\newcommand{\TWRtext}{Terminal Wealth Relative }
\newcommand{\system}[1][]{{\textit{(system #1)}}}
\newcommand{\emptyvar}{\raisebox{-0.5ex}{\scalebox{1.6}{$\cdot$}}}
\newcommand{\f}{\varphi}	
\newcommand{\vek}[1][\f]{\bm #1}												
\newcommand{\vekt}{{\bm t}_{i \emptyvar}}

\newcommand{\msupp}{\mathfrak{G}}												
\newcommand{\msuppR}{\mathfrak{R}}												
	
\newcommand{\acHPR}{Holding Period Return (HPR) }
\newcommand{\acHPRo}{Holding Period Return }
\newcommand{\acTWR}{\TWRtext\ (TWR) }

\newcommand{\HPRneu}[1][\empty]{\ifthenelse{\equal{#1}{\empty}}{\operatorname{HPR}}{\operatorname{HPR}_{#1}}} 

\newcommand{\KK}{{M}}
\newcommand{\MM}{{K}}

\usepackage{bookmark}
\hypersetup{
  pdftitle    = {general framework PT. Part II: drawdown risk measures},
  pdfauthor   = {Stanislaus Maier-Paape  and Qiji Jim Zhu},
  pdfkeywords = {},
  pdfborder={0 0 0.5}
}

\author{
  \normalsize \textsc{Stanislaus Maier-Paape $^1$ and Qiji Jim Zhu $^2$}              \\ [-0.2em]  
    \small \textit{$^1$ Institut f\"ur Mathematik, RWTH Aachen,}                      \\ [-0.5em]
    \small \textit{Templergraben 55, D-52062 Aachen, Germany}                         \\ [-0.5em]
    \small \href{mailto:maier@instmath.rwth-aachen.de}{maier@instmath.rwth-aachen.de} \\ [-0.5em]
    \small \textit{$^2$ Department of Mathematics, Western Michigan University,}      \\ [-0.5em]
    \small \textit{1903 West Michigan Avenue, Kalamazoo, Michigan, USA}               \\ [-0.5em]
    \small \href{mailto:qiji.zhu@wmich.edu}{qiji.zhu@wmich.edu}
}
\date{
  \vspace{0.25em}
  \normalsize\today
  \vspace{-1cm}
}
\title{
  \vspace{-2cm}
  \Large A General Framework for  Portfolio Theory. Part II: \\ drawdown risk measures
}

\clearscrheadfoot
\lehead{\pagemark}                   
\rohead{\pagemark}                   
\begin{document}

\maketitle

\begin{quote}
  \small
  \quoteparagraph{Abstract}
    The aim of this paper is to provide several examples of  convex risk measures necessary for the application
    of the general framework for  portfolio theory of Maier--Paape and Zhu, presented in Part I of this series \cite{maier:gfpt2017}.
    As alternative to classical portfolio risk measures such as the standard deviation we in particular construct risk measures
    related to the current drawdown of the portfolio equity. Combined with the results of Part I \cite{maier:gfpt2017},
    this allows us to calculate efficient portfolios based on a drawdown risk measure constraint.

  \quoteparagraph{Keywords} admissible convex risk measures, current drawdown, efficient frontier, portfolio theory,
                            fractional Kelly allocation, growth optimal portfolio, financial mathematics

  \quoteparagraph{AMS subject classification.} 52A41, 91G10, 91G70, 91G80, 91B30
\end{quote}



\vspace*{-0.26cm}
        \section{Introduction}    \label{sec:introduction}  

 Modern portfolio theory due to Markowitz \cite{markowitz:pfs1959} has been the state of the art in mathematical asset allocation for over 50 years.
 Recently, in Part I of this series (see Maier--Paape and Zhu \cite{maier:gfpt2017}), we generalized portfolio theory  such that 
 efficient portfolios can now be considered  for a wide range of utility functions and risk measures. The so found portfolios provide an efficient
 trade--off between utility and risk just as in the Markowitz portfolio theory. Besides the expected  return of the portfolio, which was used by
 Markowitz, now general concave utility functions are allowed, e.g. the log utility used for growth optimal portfolio theory
 (cf. Kelly \cite{kelly:nii}, Vince \cite{vince:nmm1995}, \cite{vince:mmm92}, Vince and Zhu \cite{vince:obs15}, Zhu \cite{zhu:mais07, zhu:camf12}, Hermes and Maier--Paape \cite{hermes:twr2017}).
 Growth optimal portfolios maximize the expected log returns of the portfolio yielding fastest compounded growth.

\vspace*{0.2cm}
 Besides the generalization in the utility functions, as a second breakthrough, more realistic risk measures are now allowed.
 Whereas Markowitz and also the related capital market asset pricing model (CAPM) of Sharpe \cite{sharpe:cap1964} use the standard
 deviation of the portfolio return as risk measure, the new theory of Part I in \cite{maier:gfpt2017} is applicable to a large class
 of convex risk measures.

\newpage\noindent
 The aim of this Part II is to provide and analyze several such convex risk measures related to the expected log drawdown of the portfolio returns.
 Drawdown related risk measures are believed to be superior in yielding risk averse strategies when compared to the standard deviation risk measure.
 Furthermore, empirical simulations of Maier--Paape \cite{maier:optf2015} have shown that (drawdown) risk averse strategies are also in great
 need when growth optimal portfolios are considered since using them regularly generates tremendous drawdowns  (see also van Tharp \cite{tharp:dgps08}).
 A variety of examples will be provided in Part III \cite{brenner:pmpt17}.

\vspace*{0.2cm}\noindent
 The results in this Part II are a natural generalization of Maier--Paape \cite{maier:raftf2017}, where drawdown related risk measures for a
 portfolio with only one risky asset were constructed. In that paper, as well as here, the construction of randomly drawn equity curves,
 which allows the measurement of drawdowns, is given in the framework of the growth optimal portfolio theory (see Section~\ref{sec:3} and furthermore
 Vince \cite{vince:rpm09}.) Therefore, we use Section~\ref{sec:2} to provide basics of the growth optimal theory and introduce our setup.

\vspace*{0.2cm}\noindent
 In Section~\ref{sec:4} we introduce the concept of admissible convex risk measures, discuss some of their properties and show that the ``risk part"
 of the growth optimal goal function provides such a risk measure. Then, in Section~\ref{sec:5} we apply this concept to the expected log drawdown of the portfolio
 returns. It is worth to note that some of the approximations of these risk measures yield, in fact, even positively homogeneous risk measures,
 which are strongly related to the concept of deviation measures of Rockafellar, Uryasev and Zabarankin \cite{rockafellar:mfpa2006}. According to the theory of
 Part I \cite{maier:gfpt2017} such positively homogeneous risk measures provide -- as in the CAPM model -- an affine structure of the efficient portfolios when the identity utility
 function is used. Moreover, often in this situation even a market portfolio, i.e. a purely risky efficient portfolio, related to drawdown risks
 can be provided as well.

\vspace*{0.2cm}\noindent 
Finally, note that the main Assumption~\ref{no-risk-free-investment} on the trade return matrix $T$ of \eqref{eq:returnMatrix} together with a no arbitrage
 market provides the basic market setup for application of the generalized portfolio theory of Part I \cite{maier:gfpt2017}. This is shown in the Appendix
 (Corollary~\ref{col:one-period-financial-market}). In fact, the appendix is used as a link between Part I and Part II and shows how the theory of Part I can 
 be used with risk measures constructed here. Nonetheless, Parts I and II can be read independently.

\vspace*{0.2cm}\noindent
 \textbf{Acknowledgement:}\
  We thank Ren\'e Brenner for support in generating the contour plots of the risk measures and Andreas Platen
for careful reading of an earlier version of the manuscript.


\vspace*{-0.18cm} 
               \section{Setup}  \label{sec:2}        
\vspace*{-4mm}

 For $1\leq k \leq \KK,\,\KK\in\N,$ we denote the $k$-th trading system by \system[k]. A trading system is an investment strategy applied to a
 financial instrument. Each system generates periodic trade returns, e.g. monthly, daily or the like. The net trade return of the $i$-th
 period of the $k$-th system is denoted
 by $t_{i,k}$, $1\leq i \leq N, 1\leq k\leq \KK$. Thus, we have the joint return matrix
 \begin{align*}
 \begin{tabular}{c|cccc}
    period	 & \system[1] & \system[2]	& $\cdots$	& \system[\KK]  \\
  \hline
    $1$		 & $t_{1,1}$  & $t_{1,2}$	& $\cdots$	& $t_{1,\KK}$   \\
    $2$		 & $t_{2,1}$  & $t_{2,2}$	& $\cdots$	& $t_{2,\KK}$   \\
    $\vdots$ & $\vdots$	  & $\vdots$	& $\ddots$	& $\vdots$    \\
   $N$		 & $t_{N,1}$  & $t_{N,2}$	& $\cdots$	& $t_{N,\KK}$   \\
 \end{tabular}
 \end{align*}
 and we denote
 \begin{align}\label{eq:returnMatrix}\index{historical trading returns} 
    T:=\bigg(t_{i,k}\bigg)_{\substack{1\leq i \leq N\\ 1\leq k \leq \KK}}\in\R^{N\times \KK}.
 \end{align}
\noindent
 For better readability, we define the rows of $T$, which represent the returns of the $i$-th period of our systems, as
 $$\vekt:=(t_{i,1},\dots,t_{i,\KK})\in\R^{1\times \KK}.$$

\vspace*{0.2cm}
 Following Vince \cite{vince:mmm92}, for a vector of portions $\vek:=\left(\f_1,\dots,\f_\KK\right)^\top$, where $\f_k$ stands for the portion
 of our capital invested in \system[k], we define the \acHPR of the $i$-th period as \index{Holding Period Return}
 \begin{align}\label{eq:HPRvec} 
   \HPRneu[i](\vek) := 1 + \sum\limits_{k=1}^\KK \f_k\,t_{i,k} = 1 + \langle\,\vekt^\top,{\vek}\,\rangle\,,
 \end{align}
 where $\langle \cdot,\cdot\rangle$ is the scalar product in $\R^\KK$. The \acTWR representing the gain (or loss) after the given $N$ periods,
 when the vector $\vek$ is invested over all periods, is then given as
 \begin{align*}\index{Terminal Wealth Relative!discrete}
    \TWR_N(\vek) := \prod\limits_{i=1}^N\HPRneu[i](\vek) = \prod\limits_{i=1}^N\,     \bigg(1 + \langle\,\vekt^\top,{\vek}\,\rangle\bigg)\,.
 \end{align*}

 Since a \acHPRo of zero for a single period means a total loss of our capital, we restrict $\TWR_N:\msupp\to\R$ to the domain $\msupp$
 given by the following definition:
\begin{definition}\label{def:msupp} 
  A vector of portions\ $\vek\in\R^\KK$ is called {\textit admissible} if\ $\vek\in\msupp$\ holds, where
  \begin{align}\label{eq:defmsupp} \index{admissible vector of fractions} 
   \notag \msupp &:=\Big\{\vek\in\R^\KK  \mid \HPRneu[i](\vek) \geq 0                  \quad\text{for all}\quad 1 \leq i \leq N\Big\}  \\
                 &                                                                                                                       \\
   \notag        &=\left\{\vek\in\R^\KK \mid \langle\,\vekt^\top,{\vek}\rangle \geq -1 \quad\text{for all}\quad 1\leq i\leq N\right\}.
 \end{align} 
  Moreover, we define
  \begin{align}\label{eq:msuppR} 
     \msuppR &:=\{\vek\in\msupp\mid\,\exists\,1\leq i_0\leq N\text{ s.t. }\HPRneu[i_0](\vek)=0\}.
  \end{align}
\end{definition}

\noindent
 Note that in particular $0 \in \overset{\circ}{\msupp}$ (the interior of $\msupp$) and $\msuppR = \partial\msupp$, the boundary of $\msupp$. Furthermore, negative $\f_k$
 are in principle allowed for short positions.

\begin{lemma}\label{lem:convexity} 
  The set $\msupp$ in Definition~\ref{def:msupp} is polyhedral and thus convex, as is $\overset{\circ}{\msupp}$.
\end{lemma}

\begin{proof}
	For each $i \in \gk{1,\ldots,N}$ the condition
	\begin{align*}
		\HPRneu[i](\vek)\geq 0 \quad\Longleftrightarrow\quad \langle\,\vekt^\top,{\vek}\rangle \geq\,-1
	\end{align*}
	defines a half space (which is convex). Since $\msupp$ is the intersection of a finite set of half spaces, it is itself convex, in fact even polyhedral.
	A similar reasoning yields that $\overset{\circ}{\msupp}$ is convex, too.
\end{proof}

\noindent
 In the following we denote by $\Se_1^{\KK-1} := \left\{\vek \in \R^\KK\!:\, \|\vek\| = 1\right\}$ the unit sphere in $\R^\KK$, where $\|\cdot\|$
 denotes the Euclidean norm. 

\begin{assumption}\label{no-risk-free-investment} 
  (\textbf{no risk free investment}) \\ [0.1cm]
  We assume that the trade return matrix $T$ in \eqref{eq:returnMatrix} satisfies
  \begin{align}\label{eq:asmdiscm} 
     \forall\ {\bm\theta} \in \Se_1^{\KK-1}\,\,\,\exists\ i_0 = i_0({\bm\theta}) \in \{1,\dots,N\} \quad \text{such that}\;
     \langle\,{\bm t}_{i_0\emptyvar}^\top,\,{\bm\theta}\rangle < 0.
  \end{align}
\end{assumption}

\noindent
 In other words, Assumption~\ref{no-risk-free-investment} states that no matter what ``allocation vector" ${\bm\theta} \not= 0$ is
 used, there will always be a period $i_0$ resulting in a loss for the portfolio.

\begin{remark}\label{rem:allocationVector} 
  \begin{enumerate}
    \item Since\ ${\bm\theta} \in \Se_1^{\KK-1}$ implies that\ $-{\bm\theta} \in \Se_1^{\KK-1}$, Assumption~\ref{no-risk-free-investment}
          also yields the existence of a period $j_0$ resulting in a gain for each\ ${\bm\theta} \in \Se_1^{\KK-1}$, i.e.
          \begin{align}\label{eq:remallovect} 
             \forall\ {\bm\theta} \in \Se_1^{\KK-1}\,\,\,\exists\ j_0 = j_0({\bm\theta}) \in \{1,\dots,N\}
             \quad \text{such that}\; \langle\,{\bm t}_{j_0\emptyvar}^\top,\,{\bm\theta}\rangle > 0.
          \end{align}
    \vspace*{-0.8cm}
    \item Note that with Assumption~\ref{no-risk-free-investment} automatically $\ker(T)=\{\vek[0]\}$ follows, i.e. that all
          trading systems are \textbf{linearly independent}.

    \item           It is not important whether or not the trading systems are \textbf{profitable}, since we allow short positions
         (cf. Assumption~1 in \cite{hermes:twr2017}). 
\end{enumerate}
\end{remark}

\begin{lemma}\label{lem:mboundedsupport} 
  Let the return matrix $T\in\R^{N\times \KK}$ (as in \eqref{eq:returnMatrix}) satisfy Assumption~\ref{no-risk-free-investment}. Then the set $\msupp$ in
  \eqref{eq:defmsupp} is compact.
\end{lemma}

\begin{proof}
  Since $\msupp$ is closed the lemma follows from \eqref{eq:asmdiscm} yielding $\HPR_{i_0}(s{\bm\theta}) < 0$ for $s > 0$ sufficiently large.
  Thus $\msupp$ is bounded as well.
\end{proof}


\vspace*{-0.18cm}
          \section{Randomly drawing trades}\label{sec:3}   

 Given a trade return matrix, we can construct equity curves by randomly drawing trades.

\begin{setup}\label{setup:Z} 
(\textbf{trading game})
  Assume  trading systems with trade return matrix $T$ from \eqref{eq:returnMatrix}. In a trading game the rows of $T$ are drawn randomly.
  Each row $\vekt$ has a probability of $p_i>0$, with $\sum_{i=1}^Np_i=1$. Drawing randomly and independently $\MM\in\N$ times from
  this distribution results in a  probability space
  $\Omega^{(\MM)}:=\big\{\omega=(\omega_1,\ldots,\omega_\MM)\,:\,\omega_i\in\{1,\ldots,N\}\big\}$ and a terminal wealth relative
  (for fractional trading with portion $\vek$ is used)
  \begin{align}\label{eq:twr_omega} 
   \TWR_1^\MM(\vek,\omega):=\prod_{j=1}^\MM \rk{ 1 + \langle\,{\bm t}_{\omega_j\emptyvar}^\top,{\vek}\,\rangle}\,, \quad \vek \in \overset{\circ}{\msupp}\,.
 \end{align}
\end{setup}

\vspace*{0.2cm}\noindent
 In the rest of the paper we will use the natural logarithm $\log$.
\begin{theorem}\label{theo:EZ} 
  For each\ $\vek \in \overset{\circ}{\msupp}$\ the random variable\
  $\cZ^{(\MM)}(\vek,\cdot):\,\Omega^{(\MM)} \to \R\,,\ \cZ^{(\MM)}(\vek,\omega):=\log\big(\TWR_1^\MM(\vek,\omega)\big),\; \MM \in \N$,
  has expected value
  \begin{align}\label{eq:EZ} 
    \E\left[\cZ^{(\MM)}(\vek,\cdot)\right] = \MM \cdot \log\,\Gamma(\vek),
  \end{align}
  where $\Gamma(\vek):=\prod_{i=1}^N \rk{1 + \langle\,\vekt^\top,{\vek}\,\rangle}^{p_i}$ is the weighted geometric mean of the holding period
  returns\ $\HPR_i(\vek) = 1 + \langle\,\vekt^\top,{\vek}\,\rangle > 0$\  (see \eqref{eq:HPRvec}) for all\ $\vek \in \overset{\circ}{\msupp}$.
\end{theorem}
\begin{proof} 
 For fixed\ $\MM \in \N$
 \begin{align*}
    \E\left[\cZ^{(\MM)}(\vek,\cdot)\right]
     &=\,\sum\limits_{\omega\in\Omega^{(\MM)}}\!\P\big(\{\omega\}\big)
            \ek{\log \prod\limits^{\MM}_{j=1}\,\left(1 + \langle\,{\bm t}_{\omega_j\emptyvar}^\top,{\vek}\,\rangle\right)}  \\
     &= \sum^{\MM}_{j=1}\;\;\sum\limits_{\omega\in\Omega^{(\MM)}}\!\P\big(\{\omega\}\big)\!
            \ek{\log \left(1 + \langle\,{\bm t}_{\omega_j\emptyvar}^\top,{\vek}\,\rangle\right)}
 \end{align*}
 holds. For each $j \in \gk{1,\ldots,\MM}$
 \begin{align*}
    &   \sum\limits_{\omega\in\Omega^{(\MM)}}\,\P\big(\{\omega\}\big)\,
            \ek{\log \left(1 + \langle\,{\bm t}_{\omega_j\emptyvar}^\top,{\vek}\,\rangle\right)}
      = \sum^N_{i=1} p_i \cdot \log\!\left(1 + \langle\,{\bm t}_{i\emptyvar}^\top,{\vek}\,\rangle\right)
 \end{align*}
 is independent of $j$ because each $\omega_j$ is an independent drawing. 
 We thus obtain
 \begin{align*}
     \left[\cZ^{(\MM)}(\vek,\cdot)\right]
            &= \MM \cdot \sum_{i=1}^N p_i \cdot \log\!\rk{1 + \langle\,\vekt^\top,{\vek}\,\rangle}  \\
            &= \MM \cdot \log\ek{\prod_{i=1}^N\rk{1 + \langle\,\vekt^\top,{\vek}\,\rangle}^{p_i}}
             = \MM \cdot \log\,\Gamma(\vek).
  \end{align*}
\end{proof}

 Next we want to split up the random variable $\cZ^{(\MM)}(\vek,\cdot)$ into \textbf{chance} and \textbf{risk} parts.
 Since\ $\TWR^\MM_1(\vek,\omega)>1$\ corresponds to a winning trade series\ $t_{\omega_1\emptyvar},\ldots,t_{\omega_\MM\emptyvar}$\ and \\
 $\TWR^\MM_1(\vek,\omega)<1$ analogously corresponds to a losing trade series we define the random variables corresponding to up trades
 and down trades:
\begin{definition}\label{logseries}  
  For $\vek \in \overset{\circ}{\msupp}$ we set \\ [0.2cm]
  \textbf{Up-trade log series:}
    \begin{align}\label{eq:uptrade log} 
      \cU^{(\MM)}(\vek,\omega) := \log\big(\max\{1,\TWR_1^\MM(\vek,\omega)\}\big) \ge 0.
    \end{align}
  \textbf{Down-trade log series:}
    \begin{align}\label{eq:downtrade log} 
      \cD^{(\MM)}(\vek,\omega) := \log\big(\min\{1,\TWR_1^\MM(\vek,\omega)\}\big) \le 0.
    \end{align}
\end{definition}

\noindent
 Clearly $\cU^{(\MM)}(\vek,\omega)+\cD^{(\MM)}(\vek,\omega)=\cZ^{(\MM)}(\vek,\omega)$. Hence by Theorem~\ref{theo:EZ} we get
\begin{corollary}\label{cor:EU+ED} 
  For $\vek \in \overset{\circ}{\msupp}$
  \begin{align}\label{eq:sum_down_uptrade} 
     \E\ek{\cU^{(\MM)}(\vek,\cdot)} + \E\ek{\cD^{(\MM)}(\vek,\cdot)} = \MM\cdot\log\,\Gamma(\vek)
  \end{align}
  holds.
\end{corollary}
 As in \cite{maier:raftf2017} we next search for explicit formulas for $\E\ek{\cU^{(\MM)}(\vek,\cdot)}$ and $\E\ek{\cD^{(\MM)}(\vek,\cdot)}$,
 respectively. By definition
 \begin{align}\label{eq:EU} 
   \E\ek{\cU^{(\MM)}(\vek,\cdot)} = \sum_{\omega:\TWR_1^\MM(\vek,\omega) > 1}\P\big(\{\omega\}\big) \cdot \log\Big(\TWR_1^\MM(\vek,\omega)\Big).
 \end{align}
 Assume $\omega=(\omega_1,\ldots,\omega_\MM)\in\Omega^{(\MM)}:=\{1,\ldots,N\}^\MM$ is for the moment fixed and the random variable $X_1$ counts how many
 of the $\omega_j$ are equal to $1$, i.e. $X_1(\omega)=x_1$ if in total $x_1$ of the $\omega_j$'s in $\omega$ are equal to $1$. With similar counting
 random variables $X_2,\ldots,X_N$ we obtain integer counts $x_i \ge 0$ and thus
 \begin{align}\label{eq:X_i} 
    X_1(\omega) = x_1,\ X_2(\omega)=x_2,\ \ldots,\ X_N(\omega)=x_N
 \end{align}
 with obviously $\sum_{i=1}^Nx_i=\MM$. Hence for this fixed $\omega$ we obtain
 \begin{align}\label{eq:TWR_1} 
    \TWR_1^\MM(\vek,\omega) = \prod_{j=1}^\MM \rk{1 + \langle\,{\bm t}_{\omega_j\emptyvar}^\top,{\vek}\,\rangle}
                          = \prod_{i=1}^N \rk{1 + \langle\,{\bm t}_{i\emptyvar}^\top,{\vek}\,\rangle}^{x_i}.
 \end{align}
 Therefore the condition on $\omega$ in the sum \eqref{eq:EU} is equivalently expressed as
 \begin{align}\label{eq:TWR_1_equiv} 
    \TWR_1^\MM(\vek,\omega)>1
      \;\Longleftrightarrow\; \log \TWR_1^\MM(\vek,\omega) > 0
      \;\Longleftrightarrow\; \sum_{i=1}^Nx_i \log\!\rk{1 + \langle\,{\bm t}_{i\emptyvar}^\top,{\vek}\,\rangle} > 0\,.
 \end{align}

\noindent
 To better understand the last sum, Taylor expansion may be used exactly as in Lemma~4.5 of \cite{maier:raftf2017} to obtain
\begin{lemma}\label{lem:small_f} 
  Let integers $x_i\geq 0$ with $\sum_{i=1}^Nx_i=\MM>0$ be given. Let furthermore\ $\vek = s\,{\bm\theta} \in \overset{\circ}{\msupp}$\
  be a vector of admissible portions where\ ${\bm\theta} \in \Se_1^{\KK-1}$ is fixed and $s > 0$.                                    \\ [0.1cm]
  Then there exists some $\varepsilon > 0$ (depending on\ $x_1,\ldots,x_N$ and ${\bm\theta}$) such that for all\ $s \in (0,\varepsilon]$ the
  following holds:
  \begin{enumerate}
    \item $\sum_{i=1}^N x_i \langle\,{\bm t}_{i\emptyvar}^\top,{\bm\theta}\,\rangle > 0\;\Longleftrightarrow\;h(s,{\bm\theta}) :=
           \sum_{i=1}^N x_i\,\log\!\rk{1 + s\langle\,{\bm t}_{i\emptyvar}^\top,{\bm\theta}\,\rangle} > 0$                             \\ [-0.2cm]
    \item $\sum_{i=1}^N x_i \langle\,{\bm t}_{i\emptyvar}^\top,{\bm\theta}\,\rangle \leq 0\;\Longleftrightarrow\; h(s,{\bm\theta})  =
           \sum_{i=1}^N x_i\,\log\!\rk{1 + s\langle\,{\bm t}_{i\emptyvar}^\top,{\bm\theta}\,\rangle}<0$\,
  \end{enumerate}
\end{lemma}

\begin{proof}
The conclusions  follow immediately from\ $h(0,{\bm\theta}) = 0\,,\ \frac{\partial}{\partial s}\,h(0,{\bm\theta}) = \sum\limits^N_{i=1}\,x_i
  \langle\,{\bm t}_{i\emptyvar}^\top,{\bm\theta}\,\rangle$\ and  $\frac{\partial^2}{\partial s^2}\,h(0,{\bm\theta}) < 0$.
\end{proof}

\noindent
 With Lemma~\ref{lem:small_f} we hence can restate \eqref{eq:TWR_1_equiv}. For ${\bm\theta} \in \Se_1^{\KK-1}$ and all $s \in (0,\varepsilon]$
 the following holds                                                                                                                  \vspace*{-0.3cm}
\begin{align}\label{eq:TWR_1_equiv_smallf} 
  \TWR_1^\MM(s{\bm\theta,\omega}) > 1\;\Longleftrightarrow\;\sum_{i=1}^N x_i \langle\,{\bm t}_{i\emptyvar}^\top,{\bm\theta}\,\rangle > 0.
\end{align}

\noindent
 Note that since $\Omega^{(\MM)}$ is finite and $\Se^{\KK-1}_1$ is compact, a (maybe smaller) $\varepsilon > 0$ can be
 found such that \eqref{eq:TWR_1_equiv_smallf} holds for all\ $s \in (0,\varepsilon]\,,\ {\bm\theta} \in \Se^{\KK-1}_1$\
 and $\omega \in \Omega^{(\MM)}$.

\vspace*{0.1cm}
\begin{remark}\label{rem:small f} 
  In the situation of Lemma~\ref{lem:small_f} furthermore
  \begin{align}\label{eq:lemma3.5} 
     \hspace{-3.0cm}
     \text{(b)*} \hspace{1.5cm}
           \sum_{i=1}^N x_i \langle\,{\bm t}_{i\emptyvar}^\top,{\bm\theta}\,\rangle \leq 0\quad\Longrightarrow\quad
           h(s,{\bm\theta}) < 0 \quad\text{for all}\quad s > 0\,,
  \end{align}
  holds true since $h$ is a concave function in $s$.
\end{remark}
\noindent
 After all these preliminaries, we may now state the first main result. For simplifying  the notation, we set $\N_0 := \N \cup \{0\}$ and  introduce
 \begin{align}\label{eq:preliminaries} 
    H^{(\MM,N)}\left(x_1,\ldots,x_N\right)\,:=\,p_1^{x_1}\cdots p_N^{x_N}\,\binom{\MM}{x_1\; x_2\cdots x_N}
 \end{align}
 for further reference, where\ $\displaystyle \binom{\MM}{x_1\; x_2\cdots x_N} = \frac{\MM!}{x_1!x_2!\cdots x_N!}$\ is the multinomial
 coefficient for $\left(x_1,\ldots,x_N\right) \in \N^N_0$ with\ $\sum_{i=1}^N\,x_i = \MM$\ fixed and\ $p_1,\ldots,p_N$\ are the probabilities
 from Setup~\ref{setup:Z}.
\vspace*{0.1cm}
\begin{theorem}\label{theo:EU_smallf} 
  Let a trading game as in Setup~\ref{setup:Z} with fixed\ $N,\MM \in \N$\ be given and ${\bm\theta} \in \Se_1^{\KK-1}$.
  Then there exists an $\varepsilon > 0$ such that for all $s \in (0,\varepsilon]$ the following holds:
  \begin{align}\label{eq:EU_smallf} 
     \E\ek{\cU^{(\MM)}(s{\bm\theta},\cdot)}\,=\,u^{(\MM)}(s,{\bm\theta})\,:=\,
     \sum_{n=1}^N U_n^{(\MM,N)}({\bm\theta}) \cdot \log\!\rk{1 + s\,\langle\,{\bm t}_{n\emptyvar}^\top,{\bm\theta}\,\rangle} \geq 0\,,
  \end{align}
  where
  \begin{align}\label{eq:UnMN} 
    U_n^{(\MM,N)}({\bm\theta}) := \sum_{\substack{(x_1,\ldots,x_N)\in\N_0^N  \\
                                       \sum\limits_{i=1}^N x_i = \MM,\
                                       \sum\limits_{i=1}^N x_i \langle\,{\bm t}_{i\emptyvar}^\top,{\bm\theta}\,\rangle > 0}}
    H^{(\MM,N)}\rk{x_1,\ldots,x_N} \cdot x_n \geq 0
  \end{align}
  and with\ $H^{(\MM,N)}$ from \eqref{eq:preliminaries}.
\end{theorem}
\vspace*{0.1cm}
\begin{proof}
 $\E\ek{\cU^{(\MM)}(s{\bm\theta},\cdot)} \geq 0$\ is clear from \eqref{eq:uptrade log} even for all $s \geq 0$.
  The rest of the proof is along the lines of the proof of the univariate case Theorem~4.6 in \cite{maier:raftf2017}, but will be given
  for convenience. Starting with \eqref{eq:EU} and using \eqref{eq:X_i} and \eqref{eq:TWR_1_equiv_smallf} we get for $s \in (0,\varepsilon]$
  \begin{align*}
    \E\ek{\cU^{(\MM)}(s{\bm\theta},\cdot)}\,=\,
       \sum_{\substack{(x_1,\ldots,x_N) \in \N_0^N                           \\ \sum\limits_{i=1}^N x_i\,=\,\MM} }  \hspace{2mm}
       \sum_{\substack{\omega:X_1(\omega)\,=\,x_1,\ldots,X_N(\omega)\,=\,x_N \\ \sum\limits_{i=1}^N x_i \langle\,{\bm t}_{i\emptyvar}^\top,{\bm\theta}\,\rangle > 0}}
    \P(\{\omega\})\cdot \log\!\rk{\TWR_1^\MM(s{\bm\theta},\omega)}.
  \end{align*}
\noindent
  Since there are\ $\binom{\MM}{x_1\; x_2\cdots x_N}=\frac{\MM!}{x_1!x_2!\cdots x_N!}$\ many\ $\omega\in\Omega^{(\MM)}$\ for which\
  $X_1(\omega)=x_1,\ldots,X_N(\omega)=x_N$\ holds we furthermore get from \eqref{eq:TWR_1}
  \begin{align*}
    \E\ek{\cU^{(\MM)}(s{\bm\theta},\cdot)}
      &= \sum_{\substack{(x_1,\ldots,x_N) \in \N_0^N  \\
           \sum\limits_{i=1}^N x_i\,=\,\MM,\,\sum\limits_{i=1}^N x_i \langle\,{\bm t}_{i\emptyvar}^\top,{\bm\theta}\,\rangle > 0}} \hspace{-0.5cm}
          H^{(\MM,N)}\rk{x_1,\ldots,x_N}\,\sum_{n=1}^Nx_n \cdot \log\!\rk{1 + s \cdot \langle\,{\bm t}_{n\emptyvar}^\top,{\bm\theta}\,\rangle } \\ \\
      &=\;\sum_{n=1}^N U_n^{(\MM,N)}({\bm\theta}) \cdot \log\!\rk{1 + s \cdot \langle\,{\bm t}_{n\emptyvar}^\top,{\bm\theta}\,\rangle}
  \end{align*}
  as claimed.
\end{proof}
\noindent
 A similar result holds for $\E\ek{\cD^{(\MM)}(s{\bm\theta},\cdot)}$.
\begin{theorem}\label{theo:ED_smallf} 
  We assume that the conditions of Theorem~\ref{theo:EU_smallf} hold. Then:
  \begin{enumerate}
  \item For ${\bm\theta} \in \Se_1^{\KK-1}$ and $s \in (0,\varepsilon]$
        \begin{align} \label{eq:ED_smallf} 
           \E\ek{\cD^{(\MM)}(s{\bm\theta},\cdot)} = d^{(\MM)}(s,{\bm\theta}) :=
           \sum_{n=1}^N D_n^{(\MM,N)}({\bm\theta}) \cdot \log\!\rk{1 + s\,\langle\,{\bm t}_{n\emptyvar}^\top,{\bm\theta}\,\rangle} \leq 0
        \end{align}

        holds, where
        \begin{align}\label{eq:DnMN} 
          D_n^{(\MM,N)}({\bm\theta})
           := \sum_{\substack{(x_1,\ldots,x_N)\in\N_0^N \\ \sum\limits_{i=1}^N x_i\,=\,\MM,\,
              \sum\limits_{i=1}^N x_i \langle\,{\bm t}_{i\emptyvar}^\top,{\bm\theta}\,\rangle \leq 0} }
          H^{(\MM,N)}\rk{x_1,\ldots,x_N} \cdot x_n \geq 0\,.
       \end{align}

  \item For all $s > 0$ and\ ${\bm\theta} \in \Se_1^{\KK-1}$ with\ $s{\bm\theta} \in \overset{\circ}{\msupp}$
        \begin{align}\label{eq:EDM} 
           \E\ek{\cD^{(\MM)}(s{\bm\theta},\cdot)}\,\leq\,d^{(\MM)} (s,{\bm\theta})\,\leq\,0\,,
        \end{align}
        i.e. $d^{(\MM)} (s,{\bm\theta})$ is always an upper bound for the expectation of the down--trade \\ log series.
  \end{enumerate}
\end{theorem}
\begin{remark}\label{rem:theo3.8ad-b} 
  For large $s > 0$ either\ $\E\ek{\cD^{(\MM)}(s{\bm\theta},\cdot)}$\ or\ $d^{(\MM)}(s,{\bm\theta})$\ or both shall assume
  the value $-\infty$ in case that at least one of the logarithms in their definition is not defined. Then \eqref{eq:EDM}
  holds for all\ $s{\bm\theta} \in \R^{\KK}$.
\end{remark}
\begin{proof} of Theorem~\ref{theo:ED_smallf}: \\ [0.1cm]
  \textbf{ad(a)}\;$\E\ek{\cD^{(\MM)}(s{\bm\theta},\cdot)} \leq 0$\ follows from \eqref{eq:downtrade log} again for all $s \geq 0$.
      Furthermore, \\ by definition
      \begin{align}\label{eq:TWR 1} 
         \E\ek{\cD^{(\MM)}(s{\bm\theta},\cdot)} = \sum_{\omega:\TWR_1^\MM(s{\bm\theta},\omega) < 1} \P\big(\{\omega\}\big) \cdot
         \log\!\rk{\TWR_1^\MM(s{\bm\theta},\omega)}.
      \end{align}

  The arguments given in the proof of Theorem~\ref{theo:EU_smallf} apply similarly, where instead of \eqref{eq:TWR_1_equiv_smallf} we use
  Lemma~\ref{lem:small_f} (b) to get for\ $s \in (0,\varepsilon]$
  \begin{align}\label{eq:TWR_1_equiv_smallf_b} 
    \TWR_1^\MM(s{\bm\theta},\omega) < 1\;\Longleftrightarrow\;\sum_{i=1}^N x_i \langle\,{\bm t}_{i\emptyvar}^\top,{\bm\theta}\,\rangle\,\leq\,0
  \end{align}
  for all $\omega$ with \vspace*{-0.5cm}
  \begin{align}\label{eq:omega x1} 
     X_1(\omega) = x_1\,,\ X_2(\omega) = x_2,\ldots,\,X_N(\omega) = x_N\,.
  \end{align}
\vspace*{0.15cm}
  \textbf{ad(b)}\;
      According to the extension of Lemma~\ref{lem:small_f} in Remark~\ref{rem:small f}, we also get
      \begin{align}\label{eq:ad(b)} 
         \sum\limits^N_{i=1} x_i\,\langle\,{\bm t}_{i\emptyvar}^\top,{\bm\theta}\,\rangle\,\leq\,0 \quad\Longrightarrow\quad
         \TWR_1^\MM(s{\bm\theta},\omega) < 1 \quad\text{for all}\quad s > 0
      \end{align}
      for all $\omega$ with \eqref{eq:omega x1}.
  Therefore, no matter how large $s > 0$ is, the summands of\ $d^{(\MM)}(s,{\bm\theta})$\ in \eqref{eq:ED_smallf} will always contribute
  to\ $\E\ek{\cD^{(\MM)}(s{\bm\theta},\cdot)}$\ in \eqref{eq:TWR 1}, but --- at least for large $s > 0$ --- there may be even more (negative)
  summands from other $\omega$. Hence \eqref{eq:EDM} follows for all $s > 0$.

\end{proof}
\begin{remark}\label{rem:distribution theory} 
  Using multinomial distribution theory and \eqref{eq:preliminaries}
  \begin{align*}
     \sum_{\substack{(x_1,\ldots,x_N)\in\N_0^N \\ \sum\limits_{i=1}^Nx_i\,=\,\MM}} \hspace{2mm}
     H^{(\MM,N)}\rk{x_1,\ldots,x_N}\,x_n = p_n\cdot \MM \quad \text{for all $n = 1,\ldots,N$}
  \end{align*}
\noindent
  holds and yields (again) with Theorem~\ref{theo:EU_smallf} and \ref{theo:ED_smallf} for $s \in (0,\varepsilon]$
  \begin{align*}
    \E\ek{\cU^{(\MM)}(s{\bm\theta},\cdot)}+\E\ek{\cD^{(\MM)}(s{\bm\theta},\cdot)}
      = \sum_{n=1}^N p_n \cdot \MM \cdot \log\!\rk{1 + s\langle\,{\bm t}_{n\emptyvar}^\top,{\bm\theta}\,\rangle} = \MM \cdot \log\Gamma(s{\bm\theta}).
  \end{align*}
\end{remark}
\begin{remark}\label{rem:taylor expansion}  
  Using Taylor expansion in \eqref{eq:ED_smallf} we therefore obtain a first order approximation in $s$ of the expected down-trade log series
  $\cD^{(\MM)}(s{\bm\theta},\cdot)$ (\ref{eq:downtrade log}), i.e. for $s \in (0,\varepsilon]$ and ${\bm\theta} \in \Se_1^{\KK-1}$ the
  following holds:
  \begin{align}\label{eq:logseries smallf b} 
    \E\ek{\cD^{(\MM)}(s{\bm\theta},\cdot)}\,\approx\,\tilde{d}^{(\MM)}(s,{\bm\theta})\,:=\,s\cdot\sum\limits^N_{n=1}\,D^{(\MM,N)}_n({\bm\theta})\,\cdot\,
    \langle\,{\bm t}_{n\emptyvar}^\top,{\bm\theta}\,\rangle\,.
  \end{align}
\end{remark}
 In the sequel we call $d^{(\MM)}$ the first and $\widetilde{d}^{(\MM)}$ the second approximation of the expected down--trade log series.
 Noting that\ $\log(1+x) \leq x$\ for $x \in \R$ when we extend\ $\log_{\big|_{(-\infty,0]}}\!\!:= -\infty$,\ we can
 improve part \textbf{(b)} of Theorem~\ref{theo:ED_smallf}:
\begin{corollary}\label{coro:theorem3.8} 
  In the situation of Theorem~\ref{theo:ED_smallf} for all $s \ge 0$ and\ ${\bm\theta} \in \Se_1^{\KK-1}$ such that $s{\bm\theta} \in \overset{\circ}{\msupp}$,
  we get:\vspace*{-0.18cm}
  \begin{enumerate}
  \item \begin{align}\label{eq:coro(a)}\hspace{-2.0cm} 
           \E\ek{\cD^{(\MM)}(s{\bm\theta},\cdot)} \leq d^{(\MM)}(s,{\bm\theta}) \leq \widetilde{d}^{(\MM)}(s,{\bm\theta}).
        \end{align}
  \item Furthermore $\widetilde{d}^{(\MM)}$ is continuous in $s$ and ${\bm\theta}$ (in $s$ even positive homogeneous) and
        \begin{align}\label{eq:coro(b)} 
           \widetilde{d}^{(\MM)}(s,{\bm\theta})\,\leq\,0\,.
        \end{align}
  \end{enumerate}
\end{corollary}
\begin{proof}
 \textbf{(a)}\;
  is already clear with the statement above. To show \textbf{(b)}, the continuity in $s$ of the second approximation \vspace*{-0.18cm}
  \begin{align*}
     \widetilde{d}^{(\MM)}(s,{\bm\theta})\,=\,s \cdot \sum\limits^N_{n=1}\,D^{(\MM,N)}_n({\bm\theta}) \cdot \langle\,{\bm t}_{n\emptyvar}^\top,{\bm\theta}\,\rangle\,,\; s > 0
     \end{align*}
  in \eqref{eq:logseries smallf b} is clear. But even continuity in ${\bm\theta}$ follows with a short argument:
  Using \eqref{eq:DnMN}
  \begin{align}\label{eq:using TWR_1 equivsmallf}  
     \notag \tilde{d}^{(\MM)}(s,{\bm\theta})\,
     & =\,s \cdot \sum\limits^N_{n=1}\;
      \sum_{\substack{\left(x_1,\ldots,x_N\right)\,\in\,\N_0^N  \\
            \sum\limits_{i=1}^N x_i = \MM,\ \sum\limits_{i=1}^N x_i \langle\,{\bm t}_{i\emptyvar}^\top,{\bm\theta}\,\rangle \leq 0}}        \hspace{-0.7cm}
      H^{(\MM,N)}\left(x_1,\ldots,x_N\right) \cdot x_n \cdot \langle\,{\bm t}_{n\emptyvar}^\top,{\bm\theta}\,\rangle                                      \\
     & =\,s \cdot \sum_{\substack{\left(x_1,\ldots,x_N\right)\,\in\,\N_0^N  \\
                  \sum\limits_{i=1}^N x_i = \MM,\ \sum\limits_{i=1}^N x_i \langle\,{\bm t}_{i\emptyvar}^\top,{\bm\theta}\,\rangle \leq 0}} \hspace{-0.7cm}
       H^{(\MM,N)}\left(x_1,\ldots,x_N\right) \cdot
                 {\underbrace{\sum\limits^N_{n=1}\,x_n\,\langle\,{\bm t}_{n\emptyvar}^\top,{\bm\theta}\,\rangle}_{\leq\,0}}                               \\
     &\notag
       =\,\,s \cdot \sum_{\substack{\left(x_1,\ldots,x_N\right)\,\in\,\N_0^N  \\
                    \sum\limits_{i=1}^N x_i = \MM}}, \hspace{-0.2cm} H^{(\MM,N)}\left(x_1,\ldots,x_N\right) \cdot \min\,
                    \left\{\sum\limits^N_{n=1}\,x_n\,\langle\,{\bm t}_{n\emptyvar}^\top,{\bm\theta}\,\rangle\,,0\right\}                              \\ 
       &\notag
       =:\,s \cdot L^{(\MM,N)}({\bm\theta})\,\leq\,0\,.
  \end{align}
  Since $\sum\limits^N_{n=1}\,x_n\,\langle\,{\bm t}_{n\emptyvar}^\top,{\bm\theta}\,\rangle$ is continuous in ${\bm\theta}\,,\ L^{(\MM,N)}({\bm\theta})$ is
  continuous, too, and clearly $\widetilde{d}^{(\MM)}$ is non--positive.
\end{proof}


      \section{Admissible convex risk measures}  \label{sec:4}    

 For the measurement of risk, various different approaches have been taken (see for instance \cite{foellmer:staf2002} for an introduction).
 For simplicity, we collect all for us important properties of risk measures in the following three definitions.
\begin{definition}\label{def1:admissConvex}  
  \textbf{(admissible convex risk measure)} \\ [0.1cm]
  Let $\cQ \subset \R^\KK$ be a convex set with $0 \in \cQ$. A function $\rho\!:\cQ \to \R^+_0$ is called an
  \textbf{admissible convex risk measure (ACRM)} if the following properties are satisfied:
  \begin{enumerate}
  \item $\rho(0)=0\,,\ \rho({\vek})\geq 0$\ for all\ ${\vek} \in \cQ$.
  \item $\rho$ is a convex and continuous function.
  \item For any\ ${\bm\theta} \in \Se_1^{\KK-1}$\ the function $\rho$ restricted to the set\
        $\big\{s{\bm\theta}\,:\,s > 0\big\}  \cap \cQ \subset \R^\KK$\ is strictly increasing in $s$,
        and hence in particular\ $\rho({\vek}) > 0$\ for all\ ${\vek} \in \cQ \setminus \{0\}$.
  \end{enumerate}
\end{definition}
\vspace*{0.1cm}
\begin{definition}\label{def2:admissStrict}  
  \textbf{(admissible strictly convex risk measure)} \\ [0.1cm]
  If in the situation of Definition~\ref{def1:admissConvex} the function\ $\rho\!:\,\cQ \to \R^+_0$ satisfies only
  \textit{(a)} and \textit{(b)} but is moreover strictly convex, then $\rho$ is called an
  \textbf{admissible strictly convex risk measure (ASCRM)}.
\end{definition}
\noindent
 Some of the here constructed risk measures are moreover positive homogeneous.
\begin{definition}\label{def:positivehomogeneous}  
  \textbf{(positive homogeneous)} \\ [0.1cm]
  The risk function $\rho\!:\,\R^{\KK} \to \R^+_0$ is positive homogeneous if
  \begin{align*}
     \rho(s{\vek})\,=\,s\rho({\vek})\;\;\text{for all}\;\; s > 0\;\;\text{and}\;{\vek} \in \R^{\KK}.
  \end{align*}
\end{definition}
\vspace*{0.1cm}
\begin{remark}\label{rem:AdmissStrict} 
  It is easy to see that an admissible strictly convex risk measure automatically satisfies \textit{(c)} in
  Definition~\ref{def1:admissConvex} and thus it is also an admissible convex risk measure.
  In fact, if\ $u > s > 0$\ then\ $s = \lambda u$\ for some\ $\lambda \in (0,1)$ and we obtain for ${\bm\theta} \in \Se_1^{\KK-1}$
  \begin{align*}
     \rho(s{\bm\theta}) = \rho(\lambda\,u{\bm\theta} + (1 - \lambda) \cdot 0 \cdot {\bm\theta}) \leq
     \lambda\,\rho(u{\bm\theta}) + (1 - \lambda) \rho(0 \cdot {\bm\theta}) = \lambda\,\rho(u{\bm\theta}) <
     \rho(u{\bm\theta})\,.
  \end{align*}
\end{remark}
\vspace*{0.1cm}
\begin{examples}\label{exa:ASCRM} 
  \begin{enumerate}
  \item The function $\rho_1$ with\ $\rho_1({\vek}) := {\vek}^\top \Lambda{\vek}\,,\ {\vek} \in \R^\KK$,\ for some symmetric
        positive definite matrix\ $\Lambda \in \R^{\KK \times \KK}$ is an admissible strictly
        convex risk measure (ASCRM).
  \item For a fixed vector\ $c=\left(c_1,\ldots,c_\KK\right) \in \R^\KK$, with $c_j > 0$ for $j=1,\ldots,\KK$, both,      \vspace*{-0.2cm}
       \begin{align*}
           \rho_2({\vek}) := \|{\vek}\|_{1,c}\,:=\,\sum\limits^\KK_{j=1}\,c_j|{\f}_j|\;\;\text{and}\;\;
           \rho_3({\vek}) := \|{\vek}\|_{\infty,c}\,:=\,\max\limits_{1 \leq j \leq \KK}\, \left\{c_j|{\f}_j|\right\}\,,\
       \end{align*}
        define admissible convex risk measures (ACRM).
  \end{enumerate}
\end{examples}

\noindent\vspace*{0.4cm}
 The structure of the ACRM implies nice properties about their level sets:                                                \vspace*{-0.3cm}
\begin{lemma} \label{lem:ACRM} 
  Let $\rho: \cQ \to \R^+_0$ be an admissible convex risk measure. Then the following holds:                              \vspace*{-0.2cm}
  \begin{enumerate}
   \item[\textit{(a)}]
         The set ${\cal M}(\alpha) := \gk{\vek \in \cQ: \rho(\vek) \leq \alpha}\,,\ \alpha \geq 0$, is convex
         and contains $0 \in \cQ$.
   \item[ ]\hspace{-1.0cm}\noindent
         Furthermore, if $\overline{{\cal M}(\alpha)}$ is bounded and  $\overline{{\cal M}(\alpha)} \subset \cQ$ we have: \vspace*{-0.18cm}
   \item[\textit{(b1)}]
         The boundary of ${\cal M}(\alpha)$ is characterized by
         $\,\,\partial {\cal M}(\alpha) = \gk{\vek \in \cQ: \rho(\vek) = \alpha} \not= \emptyset$.
   \item[\textit{(b2)}]
          $\partial {\cal M}(\alpha)$ is a codimension one manifold which varies  continuously in $\alpha$.
  \end{enumerate}
\end{lemma}
\begin{proof}
  ${\cal M}(\alpha)$ is a convex set, because $\rho$ is a convex function on the convex domain $\cQ$. \\
  Thus \textit{(a)} is already clear. \\ [0.1cm]
  \textit{ad\,(b):}\
    Assuming $\overline{{\cal M}(\alpha)} \subset \cQ$ is bounded immediately yields
    $\overset{\circ}{\cal M}(\alpha) = \gk{\vek \in \cQ: \rho(\vek) < \alpha}$ and
    $\partial {\cal M}(\alpha) = \gk{\vek \in \cQ: \rho(\vek) = \alpha} \not= \emptyset$, the latter being a codimension
    one manifold and continuously varying in $\alpha$ due to Definition~\ref{def1:admissConvex}\textit{(c)}.
\end{proof}

\vspace*{0.2cm}\noindent
 In order to define a nontrivial ACRM, we use the down--trade log series of \eqref{eq:downtrade log}.
\vspace*{0.1cm}
\begin{theorem}\label{theo1:ACRM} 
  For a trading game as in Setup~\ref{setup:Z} satisfying Assumption~\ref{no-risk-free-investment} the \\
  function\ $\rho_{\textnormal{down}}\!: \overset{\circ}{\msupp} \to \R^+_0$,
  \begin{align}\label{eq:varACRM} 
    \rho_{\textnormal{down}}({\vek})\,=\,\rho^{(\MM)}_{\textnormal{down}}({\vek})\,:=\,-\,\E\left(\cD^{(\MM)}({\vek},\cdot)\right) \geq 0\,,
  \end{align}
  stemming from the down--trade log series in \eqref{eq:downtrade log}, is an admissible convex risk measure (ACRM).
\end{theorem}
\begin{proof}
\; We show that $\rho_{\textnormal{down}}$ has the three properties (a), (b), and (c) from
 Definition \ref{def1:admissConvex}.
\\ [0.2cm]
  \textit{ad\,(a):}\
    $\cQ = \overset{\circ}{\msupp}$\ is a convex set with $0 \in \overset{\circ}{\msupp}$ according to Lemma~\ref{lem:convexity}. 
    Since for all $\omega \in \Omega^{(\MM)}$ and ${\vek} \in \overset{\circ}{\msupp}$
    \begin{align*}\hspace{0.6cm}
       \cD^{(\MM)}({\vek},\omega)\,=\, \log\Big(\min\left\{1,\,\TWR^\MM_1({\vek},\omega)\right\}\Big)\,=\,
                                                \min\left\{0,\,\log \TWR^\MM_1({\vek},\omega)\right\} \leq 0
    \end{align*}
    and\ $\TWR^\MM_1(0,\omega)=1$\ we obtain Definition~\ref{def1:admissConvex}\text{(a)}.
\\ [0.2cm]
  \textit{ad\,(b)}\
    For each fixed\ $\omega = \left(\omega_1,\ldots,\omega_\MM\right) \in \Omega^{(\MM)}$\ the function $ {\vek} \mapsto \TWR^\MM_1({\vek},\omega)$
    is continuous in ${\vek}$, and therefore the same holds true for $\rho_{\textnormal{down}}$. Moreover, again for $\omega \in \Omega^{(\MM)}$
    fixed,\ ${\vek} \mapsto \log\,\TWR^\MM_1({\vek},\omega) = \sum^\MM_{j=1}\,\log\!\rk{1\,+\,\langle\,{\bm t}_{\omega_j\emptyvar}^\top\,,{\vek}\rangle}$\
    is a concave function of ${\vek}$ since all summands are as composition of the concave $\log$--function with an affine function also concave.
    Thus $\cD^{(\MM)}({\vek},\omega)$ is concave as well since the minimum of two concave functions is still concave and therefore
    $\rho_{\textnormal{down}}$ is convex.
\\ [0.2cm]
  \textit{ad\,(c)}\
    It is sufficient to show that \vspace*{-0.15cm}
    \begin{align}\label{eq:thetaSC} 
       &\notag \rho_{\textnormal{down}}\ \text{from \eqref{eq:varACRM} is strictly convex along the line}\;
                              \big\{s{\bm\theta}_0:\;s > 0\big\}\,\cap\,\overset{\circ}{\msupp} \subset \R^\KK \\
       & \text{for any fixed}\ {\bm\theta}_0 \in \Se_1^{\KK-1}\,.
    \end{align}
    Therefore, let ${\bm\theta}_0 \in \Se_1^{\KK-1}$ be fixed. In order to show \eqref{eq:thetaSC} we need to find at least one
    $\overline{\omega} \in \Omega^{(\MM)}$ such that $\cD^{(\MM)}\!\left(s{\bm\theta}_0,\,\overline{\omega}\right)$ is
    strictly concave in $s > 0$.
    Using Assumption~\ref{no-risk-free-investment} we obtain some $i_0=i_0\left({\bm\theta}_0\right)$ such that\
    $\langle\,{\bm t}_{i_0\emptyvar}^\top,\,{\bm\theta}_0\,\rangle < 0$.
    Hence, for\ ${\vek}_s=s\cdot{\bm\theta}_0 \in \overset{\circ}{\msupp}$\ and\; $\overline{\omega}=\left(i_0,i_0,\ldots,i_0\right)$ we obtain
    \begin{align*}
       \cD^{(\MM)}(s{\bm\theta}_0,\overline\omega)\,=\,\MM \cdot\log \Big(1 + s\,\underbrace{\langle\,{\bm t}_{i_0\emptyvar}^\top,{\bm\theta}_0\,\rangle}_{<\,0}\,\Big) < 0
    \end{align*}
\noindent\vspace*{-0.2cm}
    which is a strictly concave function in $s > 0$.
\end{proof}
%

  \begin{figure}[htb]
    \centering
    \begin{minipage}[c]{0.95\linewidth}
        \centering
\includegraphics[width=0.9\textwidth]{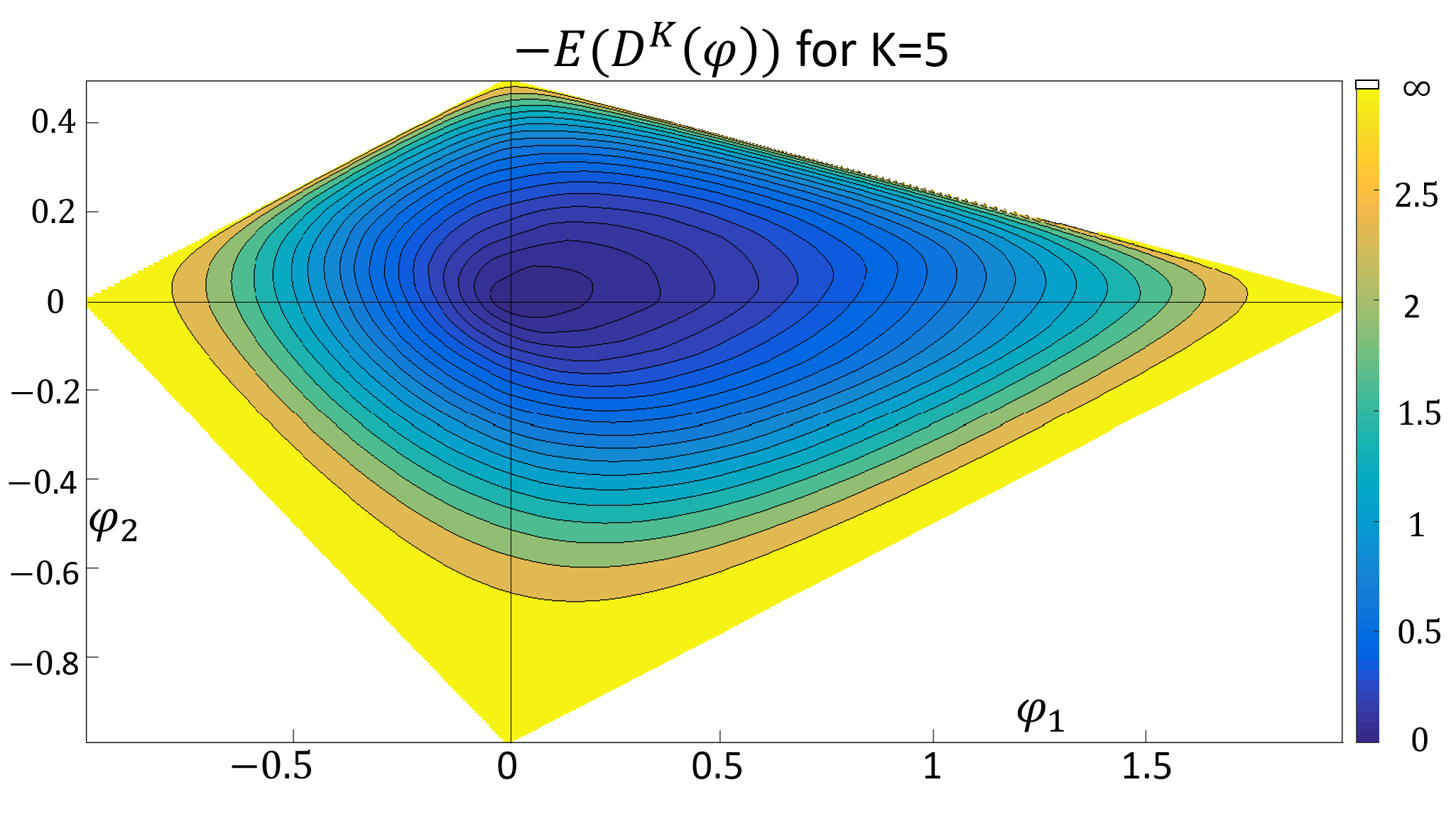}
    \end{minipage}
    \caption{Contour levels for $\rho_{\textnormal{down}}^{(K)}$ from  \eqref{eq:varACRM}
             with $\MM=5$ for $T$ from Example~\ref{exa:risk measures} }
    \label{LogDownTrades}
  \end{figure}

%
\begin{example}\label{exa:risk measures} 

  In order to illustrate $\rho_{\textnormal{down}}$ of \eqref{eq:varACRM} and the other risk measures to follow, we introduce a
  simple trading game with $M = 2$. Set
  \begin{align}\label{eq:ContourLevels} 
     T = \left(\begin{array}{r@{\hspace{0.4cm}}r} 1 & 1 \\ [0.05cm] -\,\frac{1}{2} & 1 \\ [0.05cm] 1 & -2 \\ [0.05cm] -\,\frac{1}{2} & -2 \end{array}\right) 
     \in \R^{4\times2} \quad\text{with}\quad p_1 = p_2\,=\,0.375\,,\quad p_3 = p_4 = 0.125 
  \end{align}
  It is easy to see that bets in the first system (win $1$ with probability $0.5$ or lose $-\,\frac{1}{2}$) and bets in the second system
  (win $1$ with probability $0.75$ or lose $-2$) are stochastically independent and have the same expectation value $\frac{1}{4}$.
  The contour levels of $\rho_{\textnormal{down}}$ for $K = 5$ are shown in Figure~\ref{LogDownTrades}.
\end{example}
\begin{remark}\label{rem:convex function} 
  The function $\rho_{\textnormal{down}}$ in \eqref{eq:varACRM} may or may not be an admissible strictly convex risk measure.
  To show that we give two examples: \\ [0.1cm]
  \textit{(a)}\
    For \vspace*{-0.5cm}
    \begin{align*}
       T = \begin{pmatrix} 1 & 2 \\ 2 & 1 \\ -1 & -1 \end{pmatrix} \in \R^{3\times2} \quad (N=3\,,\,\KK=2)
    \end{align*}
    the risk measure $\rho_{\textnormal{down}}$ in \eqref{eq:varACRM} for $\MM=1$ is not strictly convex. Consider for example
    $\vek_0=\alpha \cdot \binom{1}{1} \in \overset{\circ}{\msupp}$ for some fixed $\alpha >0$.
    Then for $\vek \in B_{\varepsilon}\rk{\vek_0}\,,\ \varepsilon > 0$ small, in the trading game only the third row results
    in a loss, i.e.
    \begin{align*}
       \E\rk{\cD^{(\MM=1)}(\vek,\emptyvar)} = p_3\,\log\!\rk{1 + \langle\,{\bm t}_{3\emptyvar}^\top,{\vek}\,\rangle}
    \end{align*}
    which is constant along the line $\vek_s = \vek_0 + s \cdot \binom{1}{-1} \in B_{\varepsilon}\rk{\vek_0}$
    for small $s$ and thus {\bfseries not} strictly convex.

\vspace*{0.2cm}
   \textit{(b)}\
     We refrain from giving a complete characterization for trade return matrices $T$ for which \eqref{eq:varACRM}
     results in a strictly convex function, but only note that if besides Assumption ~\ref{no-risk-free-investment} the condition                         \vspace*{-0.2cm}
     \begin{align}\label{eq:span} 
       \text{span}\gk{{\bm t}_{i \emptyvar}^\top: \langle\,{\bm t}_{i \emptyvar}^\top,{\bm\theta}\,\rangle \not =0} = \R^\KK  \quad\text{holds}
       \quad\forall\;{\bm\theta} \in \Se_1^{\KK-1}
     \end{align}
     then this is sufficient  to give strict convexity of (\ref{eq:varACRM}) and hence in this case $\rho_{\textnormal{down}}$ in \eqref{eq:varACRM}
     is actually an ASCRM.
\end{remark}
\noindent
 Now that we saw that the negative expected down--trade log series of \eqref{eq:varACRM} is an admissible convex risk measure,
 it is natural to ask whether or not the same is true for the two approximations of the expected down--trade log series
 given in \eqref{eq:ED_smallf} and \eqref{eq:logseries smallf b} as well. Starting with                                                                   \vspace*{-0.3cm}
 \begin{align*}
    d^{(\MM)}(s,{\bm\theta})\,=\,\sum\limits^N_{n=1}\,D^{(\MM,N)}_n({\bm\theta})\,\log\!\left(1 + s\,\langle\,{\bm t}_{n\emptyvar}^\top,{\bm\theta}\,\rangle\right)
 \end{align*}
 from \eqref{eq:ED_smallf}, the answer is negative. 
 The reason is simply that $D_n^{(\MM,N)}({\bm\theta})$ from \eqref{eq:DnMN} is in general not continuous for such ${\bm\theta} \in \Se_1^{\KK-1}$ for which
 $\left(x_1,\ldots,x_N\right) \in \N_0^N$ with $\sum\limits^N_{i=1}\,x_i = \MM$ exist and which satisfy\ $\sum\limits^N_{i=1}\,x_i\,\langle\,{\bm t}_{i\emptyvar}^\top,{\bm\theta}\,\rangle = 0$, but unlike in \eqref{eq:using TWR_1 equivsmallf} for $\widetilde{d}^{(\MM)}$, the sum over the log terms
 may not vanish. 
 Therefore $d^{(\MM)}(s,{\bm\theta})$ is
 in general also not continuous. A more thorough discussion of this discontinuity can be found after Theorem~\ref{theo2:ACRM}.
 On the other hand, $\widetilde{d}^{(\MM)}$ of \eqref{eq:logseries smallf b} was proved to be continuous and non--positive
 in Corollary~\ref{coro:theorem3.8}. In fact, we can obtain:

\vspace*{0.2cm}
\begin{theorem}\label{theo2:ACRM} 
  For the trading game of Setup~\ref{setup:Z} satisfying Assumption~\ref{no-risk-free-investment} the function\ $\rho_{\textnormal{down}X}\!:\,\R^\KK \to \R^+_0$,
  \vspace*{0.1cm}
  \begin{align}\label{eq:downX} 
     \rho_{\textnormal{down}X}({\vek}) = \rho^{(\MM)}_{\textnormal{down}X}(s{\bm\theta}) := -\,\tilde{d}^{(\MM)}(s,{\bm\theta}) = - s \cdot
     L^{(\MM,N)}({\bm\theta}) \geq 0\,,\; s \geq 0\; \text{and}\; {\bm\theta} \in \Se_1^{\KK-1}
  \end{align}
\noindent
  with $L^{(\MM,N)}({\bm\theta})$ from \eqref{eq:using TWR_1 equivsmallf} is an admissible convex risk measure (ACRM) according to
  Definition~\ref{def1:admissConvex} and furthermore positive homogeneous.
\end{theorem}
\begin{proof}   Clearly $\rho_{\textnormal{down}X}$ is positive homogeneous, since
  $\rho_{\textnormal{down}X}\!\left(s{\bm\theta}\right) = s\cdot\rho_{\textnormal{down}X}\!\left({\bm\theta}\right)$ for all $s \geq 0$.
  So we only need to check the (ACRM) properties. 

\vspace*{0.3cm}\noindent
  \textit{ad\,(a)} \text{\&} \textit{ad\,(b)}:\
  The only thing left to argue is the convexity of $\rho_{\textnormal{down}X}$ or the concavity of
  $\widetilde{d}^{(\MM)}\!\big(s,{\bm\theta}\big) = s \cdot L^{(\MM,N)}({\bm\theta}) \leq 0$. To see that, according to Theorem~\ref{theo:ED_smallf}
  \begin{align*}
     d^{(\MM)}(s,{\bm\theta}) = \E\ek{\cD^{(\MM)}(s{\bm\theta},\cdot)}\,, \quad\text{for}\quad {\bm\theta} \in \Se_1^{\KK-1}
     \quad \text{and}\quad s \in [0,\varepsilon]\,,
  \end{align*}
  is concave because the right hand side is concave (see Theorem~\ref{theo1:ACRM}). Hence
  \begin{align*}
     d^{(\MM)}_{\alpha}(s,{\bm\theta}) := \frac{\alpha}{\varepsilon}\,d^{(\MM)}\left(\frac{s\varepsilon}{\alpha}\,,\,{\bm\theta}\right)\,,
     \quad\text{for}\; {\bm\theta} \in \Se_1^{\KK-1}\; \text{and}\; s \in [0,\alpha]
  \end{align*}
  is also concave. Note that right from the definition of $d^{(\MM)}(s,{\bm\theta})$ in \eqref{eq:ED_smallf} and of $L^{(\MM,N)}({\bm\theta})$
  in \eqref{eq:using TWR_1 equivsmallf} it can readily be seen that for ${\bm\theta} \in \Se_1^{\KK-1}$ fixed
  \begin{align*}
     \frac{d^{(\MM)}(s,{\bm\theta})}{s} = \frac{d^{(\MM)}(s,{\bm\theta}) - d^{(\MM)}(0,{\bm\theta})}{s}\;\longrightarrow\,L^{(\MM,N)}({\bm\theta})
     \quad\text{for}\quad s \searrow 0\,.
  \end{align*} 

\noindent
  Therefore, some further calculation yields uniform convergence
  \begin{align*}
     d^{(\MM)}_{\alpha}(s,{\bm\theta})\,\longrightarrow\,s \cdot L^{(\MM,N)}({\bm\theta}) \quad\text{for}\quad \alpha \to \infty
  \end{align*}
  on the unit ball $B_1(0) := \big\{(s,{\bm\theta}): s \in [0,1]\,,\ {\bm\theta} \in \Se_1^{\KK-1}\big\}$. Now assuming $\widetilde{d}^{(\MM)}$
  being not concave somewhere, would immediately contradict the concavity of $d^{(\MM)}_{\alpha}$. 
\\ [0.2cm]
  \textit{ad\,(c)}:\
    In order to show that for any ${\bm\theta} \in \Se_1^{\KK-1}$ the function $s \mapsto \rho_{\textnormal{down}X}(s{\bm\theta}) = - s\,L^{(\MM,N)}({\bm\theta})$
    is strictly increasing in $s$, it suffices to show $L^{(\MM,N)}({\bm\theta}) < 0$. Since $L^{(\MM,N)}({\bm\theta}) \leq 0$ is
    already clear, we only have to find one negative summand in \eqref{eq:using TWR_1 equivsmallf}.
    According to Assumption~\ref{no-risk-free-investment} for all ${\bm\theta} \in \Se_1^{\KK-1}$ there is some $i_0 \leq N$ such that
    $\langle\,{\bm t}_{i_0\emptyvar}^\top,{\bm\theta}\,\rangle < 0$. Now let
    \begin{align*}\begin{array}{c}
       \left(x_1,\ldots,x_N\right)\,:=\,\left(0,\ldots,0,\,\MM,\,0,\ldots,0\right) \\  \hspace*{3.0cm}
                                                             \uparrow              \\  \hspace*{3.6cm}
                                                             \text{$i_0$--th place}
    \end{array}
    \end{align*}
    then\ $\sum\limits^N_{i=1}\,x_i\,\langle\,{\vekt^\top},{\bm\theta}\,\rangle\,=\,\MM\,\langle\,{\bm t}_{i_0\emptyvar}^\top,{\bm\theta}\,\rangle < 0$\
    giving\ $L^{(\MM,N)}({\bm\theta}) < 0$\ as claimed.
\end{proof}

\vspace*{0.3cm}\noindent
 We illustrate the contour of $\rho_{\textnormal{down}X}$ for Example~\ref{exa:risk measures} in Figure~\ref{SecondExpected-LogDownTrades}.
 As expected, the approximation of $\rho_{down}$ is best near $\vek = 0$ (cf. Figure~\ref{LogDownTrades}). 
%
  \begin{figure}[htb]
    \centering
    \begin{minipage}[c]{0.95\linewidth}
        \centering
\includegraphics[width=0.9\textwidth]{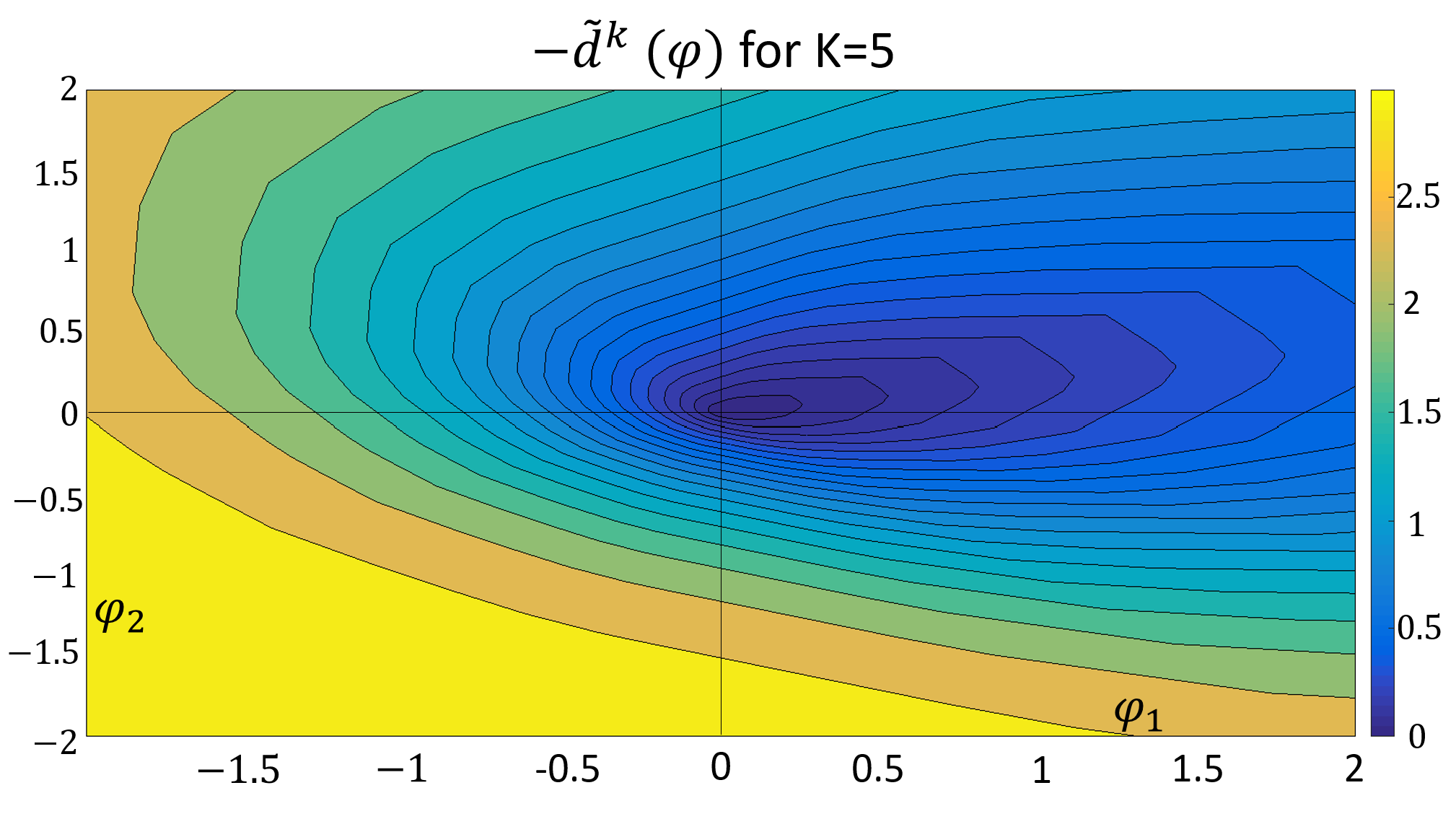}
    \end{minipage}
    \caption{Contour levels for $\rho_{downX}^{(K)}$ with $\MM=5$ for $T$ from
             Example~\ref{exa:risk measures} }
    \label{SecondExpected-LogDownTrades}
  \end{figure}


\noindent
 In conclusion, Theorems~\ref{theo1:ACRM} and \ref{theo2:ACRM} yield two ACRM stemming from expected down--trade log series $\cD^{(\MM)}$
 of \eqref{eq:downtrade log}  and its second approximation $\widetilde{d}^{(\MM)}$ from \eqref{eq:logseries smallf b}. However, the first approximation
 $d^{(\MM)}$ from \eqref{eq:ED_smallf} was not an ACRM since the coefficients\ $D_n^{(\MM,N)}$\ in \eqref{eq:DnMN} are not continuous. At first glance,
 however, this is puzzling:\ since\ $\E\left(\cD^{(\MM)}(s{\bm\theta},\cdot)\right)$ is clearly continuous and equals\ $d^{(\MM)}(s,{\bm\theta})$ for
 sufficiently small $s > 0$ according to Theorem~\ref{theo:ED_smallf}, $d^{(\MM)}(s,{\bm\theta})$ has to be continuous for small $s > 0$, too.
 So what have we missed? In order to unveil that ``mystery", we give another representation for the expected down--trade log series
 using again\ $H^{(\MM,N)}$ of \eqref{eq:preliminaries}
\vspace*{0.1cm}
\begin{lemma}\label{lem:situation-Theo3.8} 
  In the situation of Theorem~\ref{theo:ED_smallf} for all $s > 0$ and\ ${\bm\theta} \in \Se^{\KK-1}_1$\ with\
  $s{\bm\theta} \in \overset{\circ}{\msupp}$\ the following holds:
  \begin{align}\label{eq:sitTheo3.8}  
     \E\ek{\cD^{(\MM)}(s{\bm\theta},\cdot)}\,=\!\!\!\sum_{\substack{\left(x_1,\ldots,x_N\right)\in\N_0^N \\
     \sum\limits^N_{i=1} x_i = \MM}}\!\!\!H^{(\MM,N)}\left(x_1,\ldots,x_N\right) \cdot \log\!\left(\!\min\left\{1,\prod\limits^N_{n=1}\,
     \left(1 + s\langle\,{\bm t}_{n\emptyvar}^\top,{\bm\theta}\,\rangle\right)^{\!\!x_n}\!\right\}\right)\,.
  \end{align}
\end{lemma}
\begin{proof}
  \eqref{eq:sitTheo3.8} can be derived from the definition in \eqref{eq:downtrade log} as follows:
  For $\omega \in \Omega^{(\MM)}$ with \eqref{eq:X_i} clearly                                                                   \vspace*{-0.6cm}
  \begin{align*}
     \TWR^{\MM}_1(s{\bm\theta},\omega)\,=\,\prod\limits^N_{n=1}\,\left(1 + s\langle\,{\bm t}_{n\emptyvar}^\top,{\bm\theta}\,\rangle\right)^{x_n}
  \end{align*}
  holds. Introducing for $s > 0$ the set
  \begin{align}\label{eq:setXi} 
     \notag \Xi_{x_1,\ldots,x_N}(s)\,
           &:=\,\bigg\{{\bm\theta} \in \Se^{\KK-1}_1\!:
                \prod\limits^N_{j=1}\left(1 + s\langle\,{\bm t}_{j\emptyvar}^\top,{\bm\theta}\,\rangle\right)^{\!x_j} < 1\bigg\} \\
           & =\ \bigg\{{\bm\theta} \in \Se^{\KK-1}_1\!:
                \sum\limits^N_{j=1}\,x_j\,\log\!\left(1 + s\langle\,{\bm t}_{j\emptyvar}^\top,{\bm\theta}\,\rangle\right) < 0\bigg\}
 \end{align}
 and using the characteristic function of a set $A,\,\chi_A$, we obtain for all\ $s{\bm\theta} \in \overset{\circ}{\msupp}$
 \begin{align}\label{eq:4.7} 
    \E\ek{\cD^{(\MM)}(s{\bm\theta},\cdot)}\,=\!\!\!\sum_{\substack{\left(x_1,\ldots,x_N\right)\in\N_0^N \\ \sum\limits^N_{i=1} x_i = \MM}}\!\!\!
    H^{(\MM,N)}\left(x_1,\ldots,x_N\right) \cdot \chi_{\Xi_{x_1,\ldots,x_N}(s)}({\bm\theta}) \cdot
    \sum\limits^N_{n=1}\,x_n \cdot \log\!\left(1 + s\langle\,{\bm t}_{n\emptyvar}^\top,{\bm\theta}\,\rangle\right)
 \end{align}
 giving \eqref{eq:sitTheo3.8}.

\end{proof}
 Observe that\ $d^{(\MM)}(s,{\bm\theta})$ has a similar representation, namely, using
 \begin{align}\label{eq:usingXi} 
    \widehat{\Xi}_{x_1,\ldots,x_N}\,:=\,\left\{{\bm\theta} \in \Se^{\KK-1}_1\!:\,
    \sum\limits^N_{j=1}\,x_j\,\langle\,{\bm t}_{j\emptyvar}^\top,{\bm\theta}\,\rangle\,\leq\,0\right\}
  \end{align}
  we get right from the definition in \eqref{eq:ED_smallf} that for all\  $s{\bm\theta} \in \overset{\circ}{\msupp}$
  \begin{align}\label{eq:theorem3.8} 
    d^{(\MM)}(s,{\bm\theta})\,=\!\!\!\sum_{\substack{\left(x_1,\ldots,x_N\right)\in\N_0^N \\ \sum\limits^N_{i=1} x_i = \MM}}\!\!\!
    H^{(\MM,N)}\left(x_1,\ldots,x_N\right) \cdot \chi_{\widehat{\Xi}_{x_1,\ldots,x_N}}({\bm\theta}) \cdot
    \sum\limits^N_{n=1}\,x_n\,\log\!\left(1 + s\langle\,{\bm t}_{n\emptyvar}^\top,{\bm\theta}\,\rangle\right)
  \end{align}
  holds. So the only difference of \eqref{eq:4.7} and \eqref{eq:theorem3.8} is that\;$\Xi_{x_1,\ldots,x_N}(s)$\
  is replaced by\;$\widehat{\Xi}_{x_1,\ldots,x_N}$ (with the latter being a half--space restricted to\ $\Se^{\KK-1}_1$).
  Observing furthermore that due to \eqref{eq:ad(b)} \vspace*{-0.3cm}
  \begin{align}\label{eq:ad-b3.20} 
     \widehat{\Xi}_{x_1,\ldots,x_N}\,\subset\,\Xi_{x_1,\ldots,x_N}(s) \qquad \forall\ s > 0\,,
  \end{align}
  the discontinuity of $d^{(\MM)}$ clearly comes from the discontinuity of the indicator function\ $\chi_{\widehat{\Xi}_{x_1,\ldots,x_N}}$,\
  because
  \begin{align*}
     \sum\limits^N_{j=1}\,x_j \cdot \langle\,{\bm t}_{j\emptyvar}^\top,{\bm\theta}\,\rangle\,=\,0 \;\not\Longrightarrow\;
     \sum\limits^N_{n=0}\,x_n\,\ln\left(1 + s\langle\,{\bm t}_{n\emptyvar}^\top,{\bm\theta}\,\rangle\right) = 0
  \end{align*}  
  and the  ``mystery" is solved since Lemma~\ref{lem:small_f}\textbf{(b)}
  implies equality in \eqref{eq:ad-b3.20} for sufficiently small $s > 0$. Finally note that for large $s > 0$ not only the
  continuity gets lost, but moreover\ $d^{(\MM)}(s,{\bm\theta})$ is no longer concave. The discontinuity can even be seen
  in Figure~\ref{FirstExpected-LogDownTrades}.

\vspace*{-0.4cm}
  \begin{figure}[htb]
    \centering
    \begin{minipage}[c]{0.95\linewidth}
        \centering
\includegraphics[width=0.9\textwidth]{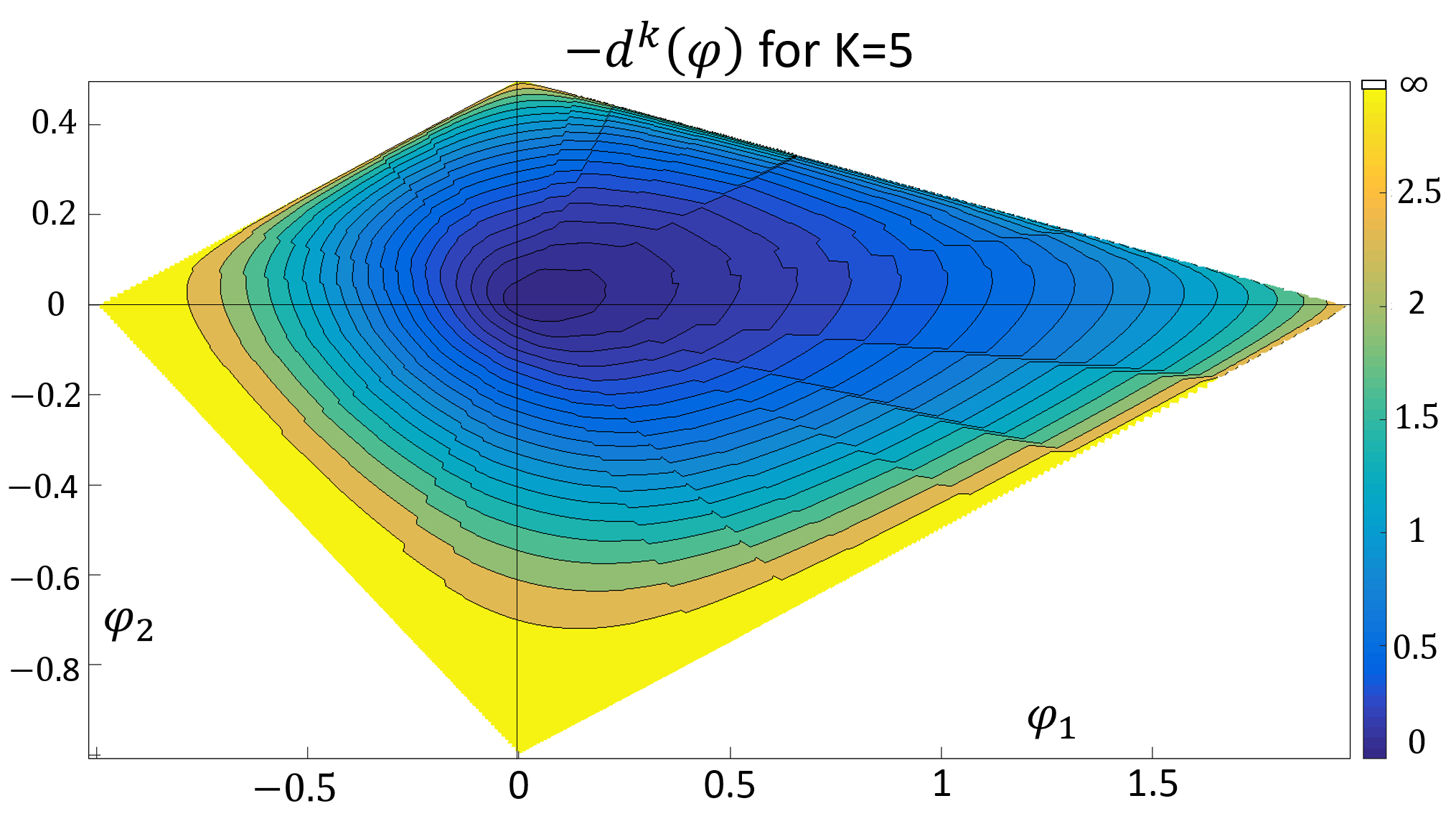}
    \end{minipage}
    \caption{Discontinuous contour levels for $-d^{(K)}$ with $\MM=5$ for $T$ from
             Example~\ref{exa:risk measures}}
    \label{FirstExpected-LogDownTrades}
  \end{figure}


\newpage\vspace*{0.1cm}
            \section{The current drawdown}  \label{sec:5}   

 We keep discussing the trading return matrix $T$ from \eqref{eq:returnMatrix} and probabilities $p_1,\ldots,p_N$ from Setup~\ref{setup:Z} for each
 row ${\vekt}$ of $T$. Drawing randomly and independently $\MM\in\N$ times such rows from that distribution results in a terminal wealth relative for
 fractional trading  \vspace*{-0.3cm}
 \begin{align*}
    \TWR_1^\MM({\vek},\omega) = \prod_{j=1}^\MM \rk{1 + \langle\,{\bm t}_{\omega_j\emptyvar}^\top\,,{\vek}\rangle}, \quad
    {\vek} \in \overset{\circ}{\msupp}\,,\ \omega\in\Omega^{(\MM)} = \gk{1,\ldots,N}^\MM\,,
 \end{align*}
 depending on the betted portions ${\vek}=\rk{\f_1,\ldots,\f_\KK}$, see \eqref{eq:twr_omega}. In order to investigate the \textit{current drawdown}
 realized after the $\MM$--th draw, we more generally use the notation
 \begin{align}\label{eq:5.1-TWR} 
    \TWR^n_m({\vek},\omega)\,:=\,\prod\limits^n_{j=m}\,\rk{1 + \langle\,{\bm t}_{\omega_j\emptyvar}^\top\,,{\vek}\rangle }\,.
 \end{align}

 The idea here is that $\TWR^n_1({\vek},\omega)$ is viewed as a discrete ``equity curve" at time $n$ (with ${\vek}$ and $\omega$ fixed).
 The current drawdown log series is defined as the logarithm of the drawdown of this equity curve realized from the maximum of the curve till
 the end (time $\MM$). We will see below that this series is the counter part of the \textit{run--up} (cf. Figure \ref{figlog-series}).


  \begin{figure}[htb]
    \centering
    \begin{minipage}[c]{0.5\linewidth}
        \centering
        \includegraphics[width=0.9\textwidth]{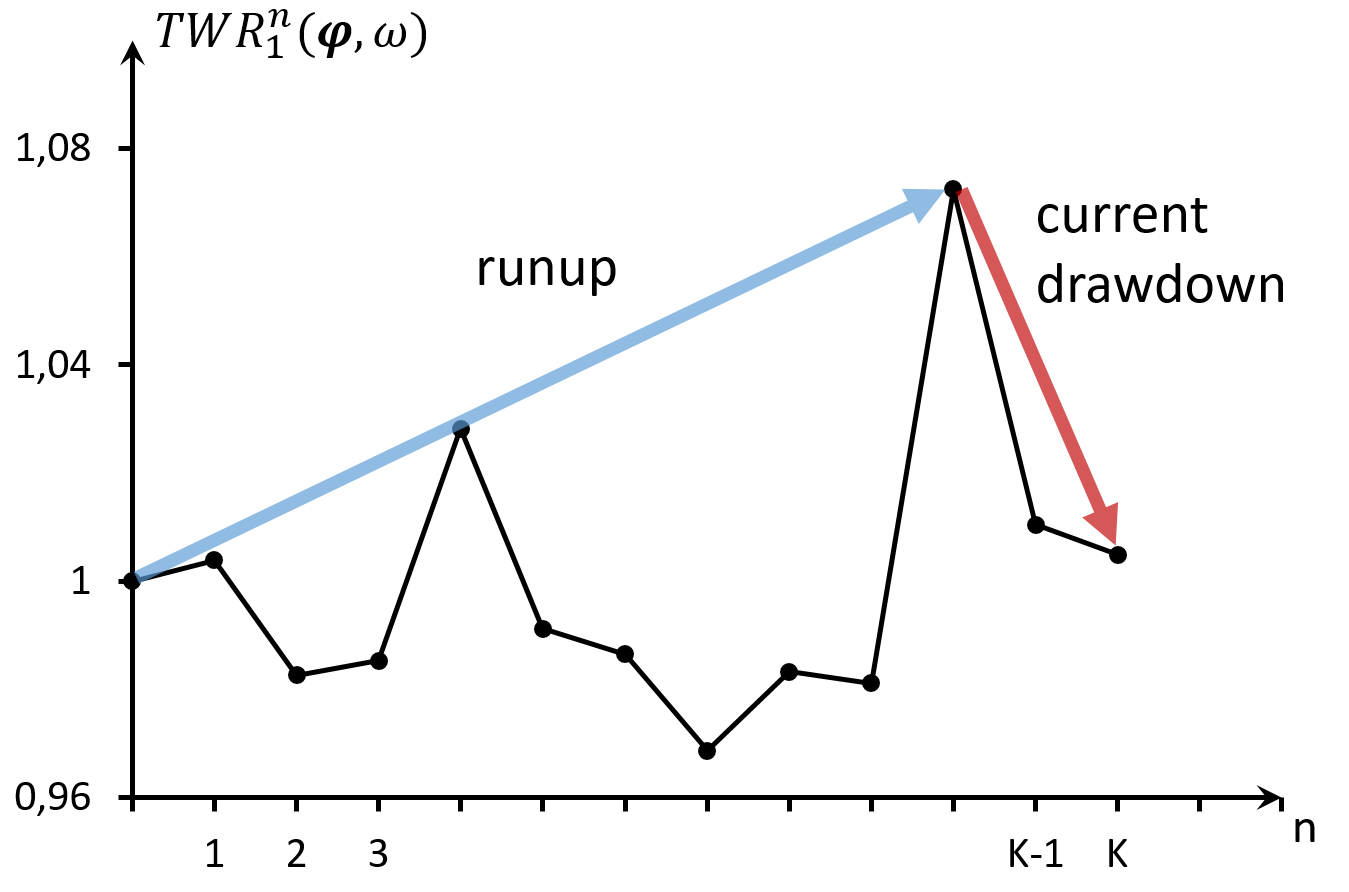}
    \end{minipage}
\hspace{-0.6cm}
    \begin{minipage}[c]{0.5\linewidth}
        \centering
       \includegraphics[width=0.9\textwidth]{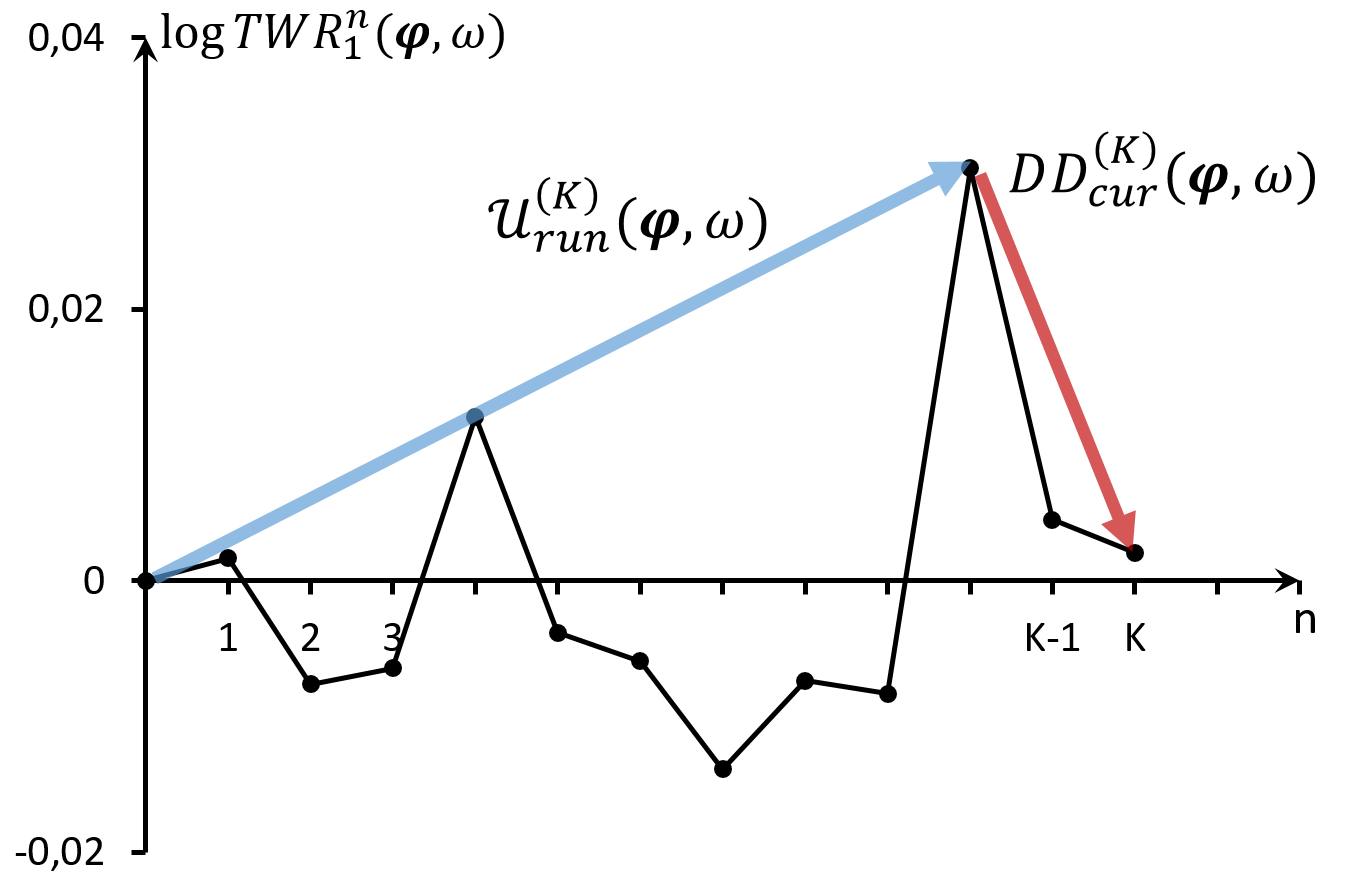}
    \end{minipage}
    \caption{In the left figure the run-up and the current drawdown is plotted for a realization
             of the \text{TWR} ``equity''--curve and to the right are their log series. }
  \label{figlog-series}
  \end{figure}

\newpage
\begin{definition}\label{def:logseries} 
  The \textbf{current drawdown log series} is set to
  \begin{align}\label{current drawdown log series} 
    \cDD^{(\MM)}({\vek},\omega) := \log\!\rk{\min_{1\,\leq\,\ell\,\leq\,\MM}\min\{1,\TWR_\ell^\MM({\vek},\omega)\}} \le 0\,,
  \end{align}
  and the \textbf{run-up log series} is defined as
  \begin{align*}
    \cU_{\textnormal{run}}^{(\MM)}({\vek},\omega) := \log\!\rk{\max_{1\,\leq\,\ell\,\leq\,\MM}\max\{1,\TWR_1^\ell({\vek},\omega)\}} \ge 0\,.
  \end{align*}
\end{definition}

 The corresponding trade series are connected because the current drawdown starts after the run--up has stopped.
 To make that more precise, we fix that $\ell$ where the run--up reached its top.
\begin{definition}\label{def:l*} 
  \textbf{(first $\TWR$ topping point)}  \\ [0.1cm]
  For fixed $\omega\in\Omega^{(\MM)}$ and ${\vek} \in \overset{\circ}{\msupp}$ define $\ell^{\ast}\!=\ell^{\ast}({\vek},\omega)\in\{0,\ldots,\MM\}$ with
  \begin{enumerate}
    \item $\ell^\ast\!= 0$\ in case\ $\max\limits_{1\,\leq\,\ell\,\leq\,\MM}\TWR_1^{\ell}({\vek},\omega) \leq 1$
    \item and otherwise choose $\ell^\ast\!\in\{1,\ldots,\MM\}$ such that
          \begin{align}\label{eq:TWR_1_l*} 
             \TWR_1^{\ell^\ast}({\vek},\omega) = \max_{1\leq\ell\leq \MM}\TWR_1^\ell({\vek},\omega) > 1\,,
          \end{align}
      where $\ell^\ast$\ should be minimal with that property.
  \end{enumerate}
\end{definition}
\noindent
 By definition one easily gets
 \begin{align}\label{eq:DD}
      \cDD^{(\MM)}({\vek},\omega) & = \begin{cases}
                                     \log\TWR_{\ell^\ast+1}^\MM({\vek},\omega), &\text{in case $\ell^\ast\!< \MM$,} \\
                                            0,                                  &\text{in case $\ell^\ast\!= \MM$},
                                  \end{cases} \\
   \intertext{and}
              \label{eq:R}
      \cU_{\textnormal{run}}^{(\MM)}({\vek},\omega)  & = \begin{cases}
                                     \log\TWR_1^{\ell^\ast}({\vek},\omega), &\text{in case $\ell^\ast\!\geq 1$,} \\
                                          0,                                &\text{in case $\ell^\ast\!= 0$}.
                               \end{cases}
 \end{align} 

\noindent
 As in Section~\ref{sec:3} we immediately obtain $\cDD^{(\MM)}({\vek},\omega) + \cU_{\textnormal{run}}^{(\MM)}({\vek},\omega) = \cZ^{(\MM)}({\vek},\omega)$
 and therefore by Theorem~\ref{theo:EZ}:
\vspace*{0.2cm}
\begin{corollary} 
  For ${\vek} \in \overset{\circ}{\msupp}$
  \begin{align}\label{eq:EDD+ER} 
    \E\ek{\cDD^{(\MM)}({\vek},\cdot)} + \E\ek{\cU_{\textnormal{run}}^{(\MM)}({\vek},\cdot)} = \MM\cdot\log\Gamma({\vek})
  \end{align}
  holds.
\end{corollary}
 Explicit formulas for the expectation of $\cDD^{(\MM)}$ and $\cU_{\textnormal{run}}^{(\MM)}$ are again of interest. \\
 By definition and with \eqref{eq:DD}
 \begin{align}\label{eq:EDD_def} 
  \E\ek{\cDD^{(\MM)}({\vek},\cdot)}
   = \sum_{\ell=0}^{\MM-1}\sum_{\substack{\omega\in\Omega^{(\MM)}\\ \ell^\ast({\vek},\omega) = \ell}}
     \P\big(\{\omega\}\big) \cdot \log\TWR_{\ell+1}^\MM\big({\vek},\omega\big).
 \end{align}
 Before we proceed with this calculation we need to discuss $\ell^\ast\!=\ell^\ast\big({\vek},\omega\big)$ further for some fixed $\omega$.
 By Definition~\ref{def:l*}, in case $\ell^{\ast}\!\ge1$, we get

 \begin{align}
        \TWR_k^{\ell^\ast}\big({\vek},\omega\big)             & > 1\quad\text{for $k=1,\ldots,\ell^\ast$,}
        \label{eq:run-up5.8)}  \\
   \intertext{since $\ell^\ast$ is the first time the run-up topped, and, in case $\ell^{\ast}\!<\MM$,}
        \TWR_{\ell^\ast+1}^{\tilde{k}}\big({\vek},\omega\big) & \leq 1\quad\text{for $\tilde{k}=\ell^\ast\!+ 1,\ldots,\MM$}.
        \label{eq:run-up5.9}                                                                                                                    \vspace*{0.3cm}
 \end{align} 

\noindent
 Similarly as in Section~\ref{sec:3} we again write ${\vek} \not= 0$ as ${\vek} = s{\bm\theta}$ for ${\bm\theta} \in \Se_1^{\KK-1}$ and $s > 0$.
 The last inequality then may be rephrased for $s \in (0,\varepsilon]$ and some sufficiently small $\varepsilon > 0$ as
 \begin{align} 
  \notag
        \TWR_{\ell^\ast+1}^{\tilde{k}}(s{\bm\theta},\omega) \leq 1
          \quad &\Longleftrightarrow \quad \log\TWR_{\ell^\ast+1}^{\tilde{k}}((s{\bm\theta},\omega) \leq 0                                                 \\
  \notag        &\Longleftrightarrow \quad \sum_{j=\ell^\ast+1}^{\tilde{k}}\log\!\rk{1 + s\,\langle\,{\bm t}_{\omega_j\emptyvar}^\top,{\bm\theta}\,\rangle} \leq 0  \\
  \label{eq:TWR-5.10}
                &\Longleftrightarrow \quad \sum_{j=\ell^\ast+1}^{\tilde{k}} \langle\,{\bm t}_{\omega_j\emptyvar}^\top,{\bm\theta}\,\rangle \leq 0
\end{align}
 by an argument similar as in Lemma~\ref{lem:small_f}. Analogously one finds for all $s \in (0,\varepsilon]$
\begin{align}\label{eq:twr k} 
   \TWR_k^{\ell^\ast}(s{\bm\theta},\omega) > 1\;\Longleftrightarrow\;\sum_{j=k}^{\ell^\ast} \langle\,{\bm t}_{\omega_j\emptyvar}^\top,{\bm\theta}\,\rangle > 0.
\end{align}
 This observation will become crucial to proof the next result on the expectation of the current drawdown.

\begin{theorem}\label{theo:EDD} 
  Let a trading game as in  Setup~\ref{setup:Z} with $N, \MM \in \N$ be fixed.
  Then for ${\bm\theta} \in \Se_1^{\KK-1}$ and $s \in (0,\varepsilon]$  the following holds:
  \begin{align}\label{eq:EDD} 
    \E\ek{\cDD^{(\MM)}(s{\bm\theta},\cdot)} = d_{\text{cur}}^{(\MM)}(s,{\bm\theta}) :=
      \sum_{n=1}^N\;\rk{\sum_{\ell=0}^\MM\Lambda_n^{(\ell,\MM,N)}({\bm\theta})} \cdot \log\!\rk{1 + s\,\langle\,{\bm t}_{n\emptyvar}^\top,{\bm\theta}\,\rangle}
  \end{align}
  where $\Lambda_n^{(\MM,\MM,N)}:=0$ is independent of ${\bm\theta}$ and for $\ell\in\{0,1,\ldots,\MM-1\}$ the functions $\Lambda^{(\ell,\MM,N)}_n({\bm\theta}) \ge 0$
  are defined by
  \begin{align}\label{eq:Lambda} 
    \Lambda_n^{(\ell,\MM,N)}({\bm\theta}) :=
      \hspace{-0.5cm}\sum_{\substack{\omega\in\Omega^{(\MM)}                                                                               \\
                     \sum\limits_{j=k}^\ell \langle\,{\bm t}_{\omega_j\emptyvar}^\top,\,{\bm\theta}\,\rangle\,>\,0\ \text{for\,$k=1,\ldots,\ell$} \\ \\
                     \sum\limits_{j=\ell+1}^{\tilde{k}} \langle\,{\bm t}_{\omega_j\emptyvar}^\top,\, {\bm\theta}\,\rangle\,\leq\,0\
                                                                                               \text{for\,$\tilde{k}=\ell+1,\ldots,\MM$}}} \hspace{-0.5cm}
    \P\big(\{\omega\}\big) \cdot \# \gk{i\,\big|\,\omega_i = n,\,i \geq \ell + 1}.
  \end{align}
\end{theorem}

\begin{proof}
  Again the proof is very similar as the proof in the univariate case, see Theorem~5.4 in \cite{maier:raftf2017}.
  Starting with \eqref{eq:EDD_def} we get

  \begin{align*}
    \E\ek{\cDD^{(\MM)}(s{\bm\theta},\cdot)} = \sum_{\ell=0}^{\MM-1} \sum_{\substack{\omega\in\Omega^{(\MM)} \\ \ell^\ast(s{\bm\theta},\omega) = \ell}}
     \P\big(\{\omega\}\big) \cdot \sum_{i=\ell+1}^\MM \log\!\rk{1 + \langle\,{\bm t}_{\omega_i\emptyvar}^\top,s{\bm\theta}\,\rangle}
  \end{align*}
  and by \eqref{eq:TWR-5.10} and \eqref{eq:twr k} for all $s \in (0,\varepsilon]$
  \begin{align}\label{bewzeile} 
    & \E\ek{\cDD^{(\MM)}(s{\bm\theta},\cdot)}\,=\,
         \sum_{\ell=0}^{\MM-1}\hspace{-0.3cm}
         \sum_{\substack{\omega\in\Omega^{(\MM)} \\
                  \sum\limits_{j=k}^\ell \langle\,{\bm t}_{\omega_j\emptyvar}^\top,\, {\bm\theta}\,\rangle\,>\,0\                \text{for\, $k=1,\ldots,\ell$}               \\
                  \sum\limits_{j=\ell+1}^{\tilde{k}} \langle\,{\bm t}_{\omega_j\emptyvar}^\top,\, {\bm\theta}\,\rangle\,\leq\,0\ \text{for\,$\tilde{k}=\ell+1,\ldots,\MM$}}}   \hspace{-1.3cm}
      \P\big(\{\omega\}\big) \cdot \sum_{i=\ell+1}^\MM\log\!\rk{1 + s\,\langle\,{\bm t}_{\omega_i\emptyvar}^\top,\, {\bm\theta}\,\rangle}
  \end{align}
\vspace*{-0.3cm}
  \begin{align*}\hspace{0.3cm}
     & =\,\sum_{\ell=0}^{\MM-1}\hspace{-0.3cm}
          \sum_{\substack{\omega\in\Omega^{(\MM)} \\
                  \sum\limits_{j=k}^\ell \langle\,{\bm t}_{\omega_j\emptyvar}^\top,\, {\bm\theta}\,\rangle\,>\,0\                \text{for\,$k=1,\ldots,\ell$}                \\
                  \sum\limits_{j=\ell+1}^{\tilde{k}} \langle\,{\bm t}_{\omega_j\emptyvar}^\top,\, {\bm\theta}\,\rangle\,\leq\,0\ \text{for\,$\tilde{k}=\ell+1,\ldots,\MM$}}}   \hspace{-1.3cm}
       \P\big(\{\omega\}\big) \cdot \sum_{n=1}^N \# \gk{i\,\big|\,\omega_i = n,\,i \geq \ell + 1} \cdot \log\!\rk{1 + s\,\langle\,{\bm t}_{n\emptyvar}^\top,{\bm\theta}\,\rangle} \\ \\
     & =\,\sum_{n=1}^N\;\;\sum_{\ell=0}^{\MM-1}\,\Lambda_n^{(\ell,\MM,N)}({\bm\theta}) \cdot \log\!\rk{1 + s\,\langle\,{\bm t}_{n\emptyvar}^\top,{\bm\theta}\,\rangle}\,
       =\,d^{(\MM)}_{\textnormal{cur}}(s,{\bm\theta})
  \end{align*}

\noindent
 since $\Lambda^{(\MM,\MM,N)}_n = 0$.
\end{proof}

\noindent
 In order to simplify notation, we introduce formally the ``linear equity curve" for \\
 $1 \leq m \leq n \leq K\,,\ \omega \in \Omega^{(\MM)} = \big\{1,\ldots,N\big\}^{\MM}$\ and\ ${\bm\theta} \in \Se^{\KK-1}_1$:
 \begin{align}\label{eq:lin-equity-curve} 
    \linEQ^n_m({\bm\theta},\omega)\,:=\,\sum\limits^n_{j=m}\ \langle\,{\bm t}_{\omega_j\emptyvar}^\top,\, {\bm\theta}\,\rangle
 \end{align}
 Then we obtain similarly to the first topping point $\ell^{\ast} = \ell^{\ast}({\vek},\omega)$ of the $\TWR$--equity curve \eqref{eq:5.1-TWR}
 (cf. Definition~\ref{def:l*}) a first topping point for the linear equity:
\begin{definition}\label{def:first linEq} 
  \textbf{(first linear equity topping point)} \\ [0.1cm]
  For fixed\ $\omega \in \Omega^{(\MM)}$ and ${\bm\theta} \in \Se^{\KK-1}_1$\ define\ $\Lhat^{\ast}=\Lhat^{\ast}({\bm\theta},\omega) \in \{0,\ldots,\MM\}$ with
  \begin{enumerate}
  \item $\Lhat^{\ast} = 0$ in case $\max\limits_{1\,\leq\,\ell\,\leq\,\MM}\ \linEQ^{\ell}_1({\bm\theta},\omega) \leq 0$
  \item and otherwise choose\ $\Lhat^{\ast} \in \{1,\ldots,\MM\}$ such that
        \begin{align}\label{eq:def5.9(b)} 
           \linEQ^{\Lhat^{\ast}}_1({\bm\theta},\omega)\,=\,
           \max\limits_{1\,\leq\,\ell\,\leq\,\MM}\ \linEQ^{\ell}_1({\bm\theta},\omega) > 0\,,
        \end{align}
        where $\Lhat^{\ast}$ should be minimal with that property.
  \end{enumerate}
\end{definition}
\vspace{0.1cm}\noindent
 Let us discuss\ $\Lhat^{\ast}=\Lhat^{\ast}({\bm\theta},\omega)$ further for some fixed $\omega$. By Definition~\ref{def:first linEq}, in case
 $\Lhat^{\ast} \geq 1$, we get
 \begin{align}\label{eq:5.21-linEquity} 
     \linEQ^{\Lhat^{\ast}}_k({\bm\theta},\omega) > 0 \quad\text{for}\quad k=1,\ldots,\Lhat^{\ast}
 \end{align}
 since $\Lhat^{\ast}$ is the first time the run-up of the linear equity topped and, in case $\Lhat^{\ast} < \MM$
 \begin{align}\label{eq:5.22-linEquity} 
    \linEQ^{\widetilde{k}}_{\Lhat^{\ast}+1}({\bm\theta},\omega) \leq 0 \quad\text{for}\quad \widetilde{k} = \Lhat^{\ast} + 1,\ldots,\MM\,.
 \end{align}
\noindent
 Hence we conclude that\ $\omega \in \Omega^{(\MM)}$ satisfies\ $\Lhat^{\ast}({\bm\theta},\omega) = \ell$\ if and only if
 \begin{align}\label{eq:ToppingPoint} 
    \sum\limits_{j=k}^\ell \langle\,{\bm t}_{\omega_j\emptyvar}^\top,\, {\bm\theta}\,\rangle\,>\,0\                \text{for\,$k=1,\ldots,\ell$}
    \quad\text{and}\quad
    \sum\limits_{j=\ell+1}^{\tilde{k}} \langle\,{\bm t}_{\omega_j\emptyvar}^\top,\, {\bm\theta}\,\rangle\,\leq\,0\ \text{for\,$\tilde{k}=\ell+1,\ldots,\MM$}\,.
 \end{align}
 Therefore \eqref{eq:Lambda} simplifies to
 \begin{align}\label{eq:Lambda5.13} 
   \Lambda_n^{(\ell,\MM,N)}({\bm\theta})\,=\!\!\!\sum\limits_{\substack{\omega\in\Omega^{(\MM)} \\ \Lhat^{\ast}({\bm\theta},\omega) = \ell}}
   \P\big(\{\omega\}\big) \cdot \# \gk{i\,\big|\,\omega_i = n,\,i \geq \ell + 1}\,.
 \end{align}

\noindent
 Furthermore, according to \eqref{eq:TWR-5.10} and \eqref{eq:twr k}, for small $s > 0$\,,\ $\ell^{\ast}\, \text{ad}\ \Lhat^{\ast}$
 coincide, i.e.
 \begin{align}\label{eq:TWR5.10-5.11} 
    \Lhat^{\ast}({\bm\theta},\omega)\,=\,\ell^{\ast}(s{\bm\theta},\omega) \quad\text{for all}\quad s \in (0,\varepsilon]\,.
 \end{align}

\noindent
 A very similar argument as the proof of Theorem~\ref{theo:EDD} yields:
\begin{theorem}\label{theo:ER} 
  In the situation of Theorem~\ref{theo:EDD} for ${\bm\theta} \in \Se_1^{\KK-1}$ and all $s \in (0,\varepsilon]$
  \begin{align}\label{eq:ER} 
     \E\ek{\cU_{\textnormal{run}}^{(\MM)}(s{\bm\theta},\cdot)} = u_{\textnormal{run}}^{(\MM)}(s,{\bm\theta}) :=
         \sum_{n=1}^N\,\rk{\sum_{\ell=0}^\MM \Upsilon_n^{(\ell,\MM,N)}({\bm\theta})} \cdot \log\!\rk{1 + s\, \langle\,{\bm t}_{n\emptyvar}^\top,{\bm\theta}\,\rangle}
  \end{align}
  holds, where $\Upsilon_n^{(0,\MM,N)}:=0$ is independent from ${\bm\theta}$ and for $\ell\in\{1,\ldots,\MM\}$ the functions
  $\Upsilon_n^{(\ell,\MM,N)}({\bm\theta}) \ge 0$ are given as
  \begin{align}\label{eq:R_n^lMN} 
    \Upsilon_n^{(\ell,\MM,N)}({\bm\theta}) := \!\!\!\sum_{\substack{\omega\in\Omega^{(\MM)} \\ \Lhat^{\ast}({\bm\theta},\omega) = \ell}}
     \P\big(\{\omega\}\big)\cdot \# \gk{i\,\big|\,\omega_i=n\,,\,i\leq\ell}\,.
  \end{align} 
\end{theorem}
\begin{remark}\label{rem:drawdownlog-series} 
  Again, we immediately obtain a first order approximation for the expected current drawdown log series. For $s \in (0,\varepsilon]$
  \begin{align}\label{eq:log-series} 
      \E\ek{\cDD^{(\MM)}(s{\bm\theta},\cdot)}\,\approx\,\tilde{d}_{\textnormal{cur}}^{(\MM)}(s,{\bm\theta})\,:=\, s \cdot \sum\limits^N_{n=1}\,
      \rk{\sum\limits^\MM_{\ell=0}\ \Lambda_n^{\rk{\ell,\MM,N}}({\bm\theta})} \cdot \langle\,{\bm t}_{n\emptyvar}^\top,{\bm\theta}\,\rangle
  \end{align}
  holds. Moreover, since
  $\cDD^{(\MM)}({\vek},\omega) \leq \cD^{(\MM)}({\vek},\omega) \leq 0\,,\ d_{\textnormal{cur}}^{(\MM)}(s,{\bm\theta}) \leq d^{(\MM)}(s,{\bm\theta}) \leq 0$
  and $\tilde{d}_{\textnormal{cur}}^{(\MM)}(s,{\bm\theta}) \leq \tilde{d}^{(\MM)}(s,{\bm\theta}) \leq 0$ holds as well.
\end{remark}
 As discussed  in Section \ref{sec:4} for the down-trade log series,  we also want to study the current drawdown log series
 (\ref{current drawdown log series}) with respect to admissible convex risk measures.
\vspace*{0.3cm}
\begin{theorem}\label{theo4:ACRM} 
  For a trading game as in Setup~\ref{setup:Z} satisfying Assumption~\ref{no-risk-free-investment} the \\ function
  \begin{align}\label{eq:varACRM4} 
    \rho_{\textnormal{cur}}({\vek})\,=\,\rho^{(\MM)}_{\textnormal{cur}}({\vek})\,:=\,-\,\E\ek{\cD_{\textnormal{cur}}^{(\MM)}({\vek},\cdot)}
    \geq 0\,,\quad {\vek} \in \overset{\circ}{\msupp}\,,
  \end{align}
  is an admissible convex risk measure (ACRM).
\end{theorem}
\begin{proof}
  It is easy to see that the proof of Theorem \ref{theo1:ACRM} can almost literally be adapted to the current drawdown case.

\end{proof}


  \begin{figure}[htb]
    \centering
    \begin{minipage}[c]{0.95\linewidth}
        \centering
\includegraphics[width=0.9\textwidth]{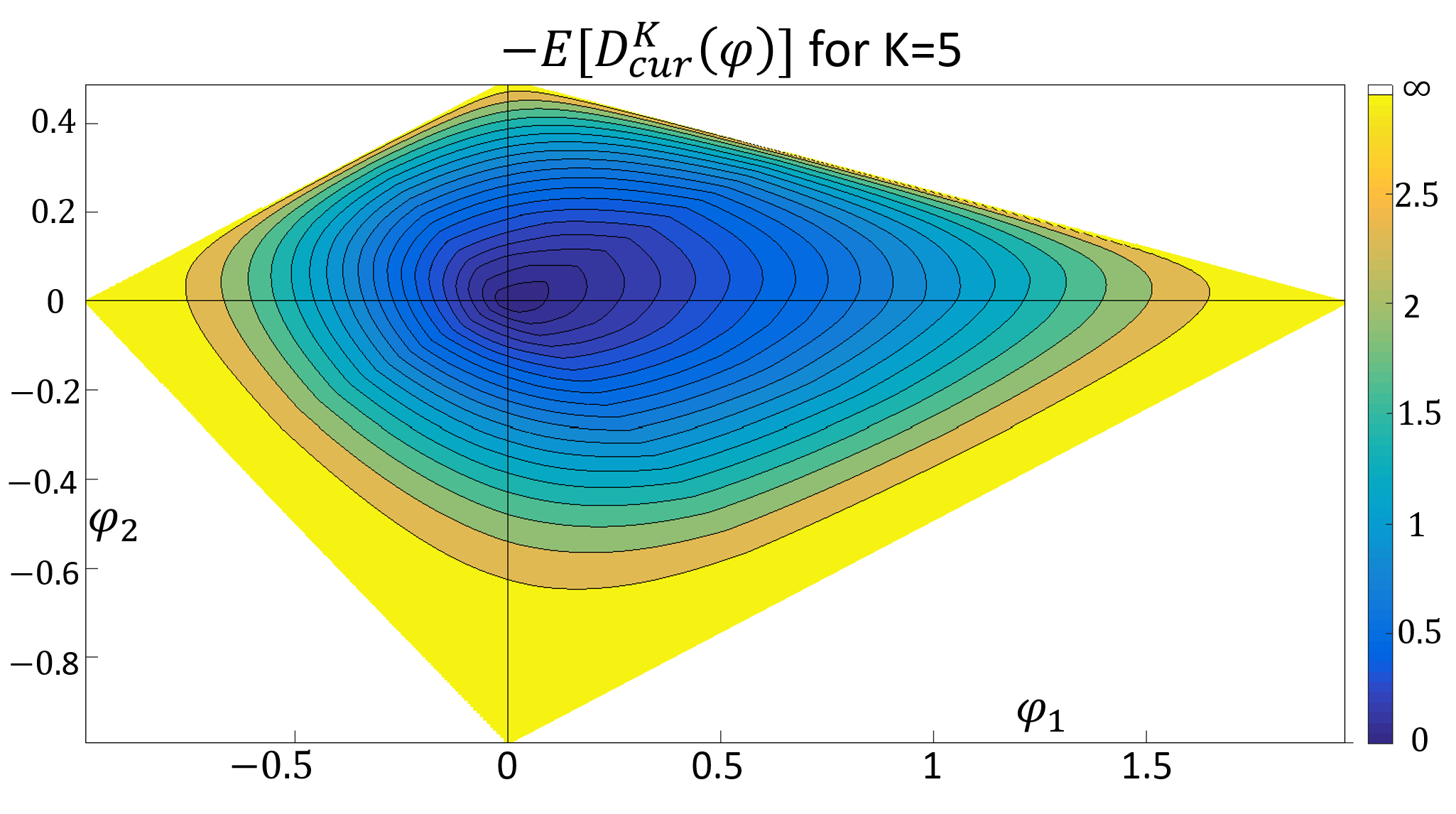}
    \end{minipage}
    \caption{Contour levels for $\rho^{(K)}_{cur}$ from \eqref{eq:varACRM4} with $\MM=5$ for
             Example~\ref{exa:risk measures} }
    \label{figExpected-LogDownTrades}
  \end{figure}


\noindent
 Confer Figure~\ref{figExpected-LogDownTrades} for an illustration of $\rho_{\textnormal{cur}}$. Compared to $\rho_{\textnormal{down}}$ in
 Figure~\ref{LogDownTrades} the contour plot looks quite similar, but near $0 \in \R^{\KK}$  obviously $\rho_{\textnormal{cur}}$ grows faster.
 Similarly, we obtain an ACRM for the first order approximation $\tilde{d}_{\textnormal{cur}}^{(\MM)}(s,{\bm\theta})$ in (\ref{eq:log-series}):         \newpage
\begin{theorem}\label{theo5:ACRM} 
  For the trading game of Setup~\ref{setup:Z} satisfying Assumption~\ref{no-risk-free-investment} the function\
  $\rho_{\textnormal{cur}X}\!:\,\R^\KK \to \R^+_0$,
  \begin{align*}
     \rho_{\textnormal{cur}X}({\vek}) = \rho^{(\MM)}_{\textnormal{cur}X}(s{\bm\theta})\,:=\,-\,\tilde{d}_{\textnormal{cur}}^{(\MM)}(s,{\bm\theta})\,
     =\,-\,s \cdot L_{\textnormal{cur}}^{(\MM,N)}({\bm\theta}) \geq 0\,,\; s \geq 0 \quad\text{and}\quad {\bm\theta} \in \Se_1^{\KK-1}
  \end{align*}
  with
  \begin{align}\label{eq:theorem5.9}  
     &L_{\textnormal{cur}}^{(\MM,N)}({\bm\theta})\,:=\,
         \sum_{\ell=0}^{\MM-1}\;
         \sum_{\substack{\omega\in\Omega^{(\MM)} \\ \Lhat^{\ast}({\bm\theta},\omega) = \ell}}
      \P\big(\{\omega\}\big) \cdot \sum_{i=\ell+1}^\MM\,\langle\,{\bm t}_{\omega_i\emptyvar}^\top,\, {\bm\theta}\,\rangle
  \end{align}
  is an admissible convex risk measure (ACRM) according to Definition~\ref{def1:admissConvex} which is moreover positive homogeneous.
\end{theorem}

\begin{proof}
  We use (\ref{bewzeile})  to derive the above formula for $L_{\textnormal{cur}}^{(\MM,N)}({\bm\theta})$.
  Now most of the arguments of the proof of Theorem~\ref{theo2:ACRM} work here as well once we know that  $L_{\textnormal{cur}}^{(\MM,N)}({\bm\theta})$
  is continuous in ${\bm\theta}$. To see that, we remark once more that for the first topping point
  $\widehat\ell^{\ast}\!=\widehat\ell^{\ast}({\bm\theta},\omega)\in\{0,\ldots,\MM\}$ of the linearized equity curve
  $\sum_{j=1}^n     \,\langle\,{\bm t}_{\omega_j\emptyvar}^\top,\, {\bm\theta}\,\rangle, \quad n=1, \ldots, \MM$, the following holds
  (cf. Definition~\ref{def:first linEq} and \eqref{eq:5.22-linEquity}):
  \begin{align*}
     \linEQ^{\MM}_{\Lhat^{\ast}+1}({\bm\theta},\omega)\,=\,\sum\limits^{\MM}_{i=\Lhat^{\ast}+1}\,
     \langle\,{\bm t}_{\omega_i\emptyvar}^\top,\, {\bm\theta}\,\rangle\,\leq\,0\,.
  \end{align*}

\noindent 
 Thus
  \begin{align*}
    &    L_{\textnormal{cur}}^{(\MM,N)}({\bm\theta})\,=\,
         \sum_{\ell=0}^{\MM-1}
         \sum_{\substack{\omega\in\Omega^{(\MM)} \\
      \hat\ell^{\ast}({\bm\theta},\omega)=\ell}}
      \P\big(\{\omega\}\big) \cdot \underbrace{\sum_{i=\ell+1}^\MM\,\langle\,{\bm t}_{\omega_i\emptyvar}^\top,\, {\bm\theta}\,\rangle}_{\le 0} \le 0\,.
\end{align*} 

\noindent
 Although the topping point $\hat\ell^{\ast}({\bm\theta},\omega)$ for $\omega \in \Omega^{(\MM)}$ may jump when $\bm\theta$ is varied in case\
 $\sum_{i\,=\,\Lhat^{\ast}+1}^j\,\langle\,{\bm t}_{\omega_i\emptyvar}^\top,\, {\bm\theta}\,\rangle=0$ for some\ $j \geq \Lhat^{\ast} + 1$, i.e.
 \begin{align*}
    \sum\limits^{\MM}_{i=\Lhat^{\ast}+1}\,\langle\,{\bm t}_{\omega_i\emptyvar}^\top,\,{\bm\theta}\,\rangle\,=\,
    \sum\limits^{\MM}_{i=j}\,\langle\,{\bm t}_{\omega_i\emptyvar}^\top,\,{\bm\theta}\,\rangle\,,
 \end{align*}
 the continuity of $L_{\textnormal{cur}}^{(\MM,N)}({\bm\theta})$ is still granted since over all $\ell =0,\ldots,\MM-1$  is summed. Hence, all claims are proved.

\end{proof}

\vspace*{0.3cm}

  \begin{figure}[htb]
    \centering
    \begin{minipage}[c]{0.95\linewidth}
        \centering
\includegraphics[width=0.9\textwidth]{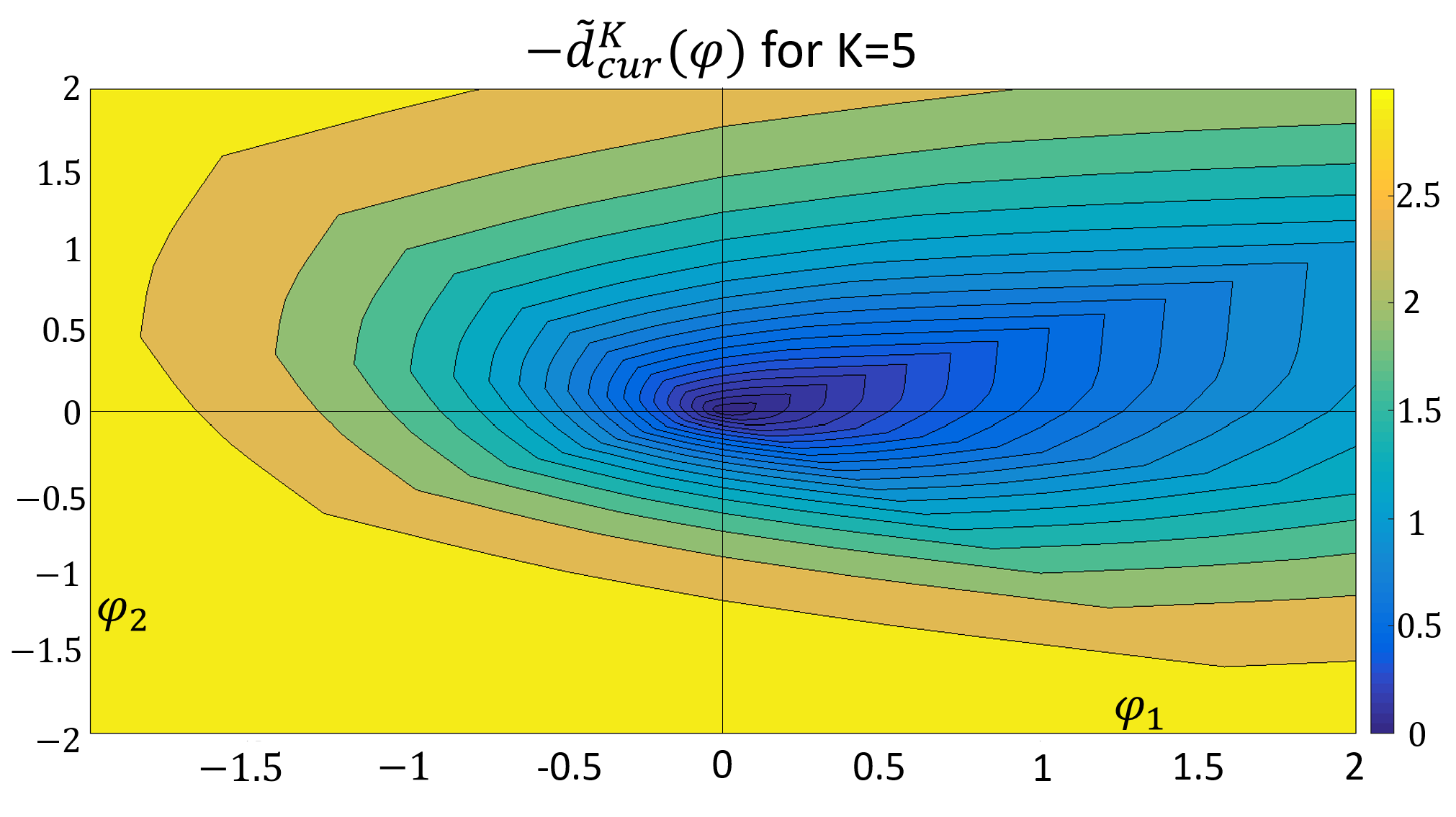}
    \end{minipage}
    \caption{Contour levels for $\rho^{(K)}_{curX}$ from Theorem~\ref{theo5:ACRM} with $\MM=5$ for
             Example~\ref{exa:risk measures} }
    \label{figSecond-LogDrawDown}
  \end{figure}


 A contour plot of $\rho_{\textnormal{cur}X}$ can be seen in Figure~\ref{figSecond-LogDrawDown}. The first topping point of the linearized equity
 curve will also be helpful  to order the risk measures $\rho_{\textnormal{cur}}$ and $\rho_{\textnormal{cur}X}$.
 Reasoning as in \eqref{eq:TWR-5.10} (see also Lemma~\ref{lem:small_f}) and using that \eqref{eq:5.22-linEquity} we obtain in case $\Lhat^{\ast}\!<\MM$
 for $s \in (0,\varepsilon]$ and\ $\widetilde{k} = \Lhat^{\ast} + 1,\ldots,\MM$ that
 \begin{align}\label{eq:5.23-linEquity} 
    \linEQ^{\widetilde{k}}_{\Lhat^{\ast}+1}({\bm\theta},\omega)\,=\,\sum\limits^{\widetilde{k}}_{j=\Lhat^{\ast}+1}\,
   \langle\,{\bm t}_{\omega_j\emptyvar}^\top,\, {\bm\theta}\,\rangle\,\leq\,0\;\Longrightarrow\;
    \sum\limits^{\widetilde{k}}_{j=\Lhat^{\ast}+1}\,\log\!\left(1\,+\,s\,\langle\,{\bm t}_{\omega_j\emptyvar}^\top,\, {\bm\theta}\,\rangle\right) \leq 0\,.
\end{align}

\noindent
 However, since $\log$ is concave, the above implication holds true even for all $s > 0$ with\ ${\vek} = s{\bm\theta} \in \overset{\circ}{\msupp}$.\
 Hence for\ $\widetilde{k} = \Lhat^{\ast} + 1,\ldots,\MM$\ and\ ${\vek} = s{\bm\theta} \in \overset{\circ}{\msupp}$
 \begin{align}\label{eq:5.24-logTWR} 
    \linEQ^{\widetilde{k}}_{\Lhat^{\ast}+1}({\bm\theta},\omega) \leq 0\;\Longrightarrow\;
  \log\TWR^{\widetilde{k}}_{\Lhat^{\ast}+1} (s{\bm\theta},\omega) \leq 0\,.
 \end{align}

\noindent
 Looking at \eqref{eq:run-up5.9} once more, we observe that the first topping point of the $\TWR$ equity curve $ \ell^{\ast}$ necessarily is
 less than or equal to $\Lhat^{\ast}$. Thus we have shown:
\begin{lemma}\label{lem:FirstTopPoint} 
  For all\ $\omega \in \Omega^{(\MM)}$ and ${\vek} = s{\bm\theta} \in \overset{\circ}{\msupp}$\ the following holds
  $\big($see also \eqref{eq:TWR5.10-5.11}$\big)$:
  \begin{align}\label{eq:5.25-1topPoint} 
     \ell^{\ast}(s{\bm\theta},\omega)\,\leq\,\Lhat^{\ast}({\bm\theta},\omega)\,.
  \end{align} 
\end{lemma}

\noindent
 This observations helps to order\ $\E\ek{\cD_{\textnormal{cur}}^{(\MM)}(s{\bm\theta},\emptyvar)}$ and\ $d^{(\MM)}_{\textnormal{cur}}(s,{\bm\theta})$:
\begin{theorem}\label{theo:5.11-observation} 
  For all\ ${\vek} = s{\bm\theta} \in \overset{\circ}{\msupp}$,\ with $s > 0$ and\ ${\bm\theta} \in \Se^{\KK-1}_1$\ we have
 \begin{align}\label{eq:5.26-observation} 
     \E\ek{\cD_{\textnormal{cur}}^{(\MM)}(s{\bm\theta},\emptyvar)}\,\leq\,d^{(\MM)}_{\textnormal{cur}}(s,{\bm\theta})\,\leq\,
    \widetilde{d}^{(\MM)}_{\textnormal{cur}}(s,{\bm\theta})\,\leq\,0\,.
  \end{align}
\end{theorem}

\begin{proof}
  Using \eqref{eq:EDD_def} for ${\vek}=s{\bm\theta} \in \overset{\circ}{\msupp}$
 \begin{align*}\begin{array}{rc@{\hspace{0.2cm}}l}
     \displaystyle \E\ek{\cD_{\textnormal{cur}}^{(\MM)}(s{\bm\theta},\emptyvar)}\,&=&\displaystyle \sum_{\ell=0}^{\MM-1}\;\sum_{\substack{\omega\in\Omega^{(\MM)} \\
                    \ell^{\ast}(s{\bm\theta},\omega)=\ell}}
                    \P\big(\{\omega\}\big) \cdot \log\TWR^{\MM}_{\ell + 1}(s{\bm\theta},\omega)\,.
\\ [1.0cm]
  & \overset{\text{\tiny Lemma~\ref{lem:FirstTopPoint}}}{\leq}
      &\displaystyle \sum_{\ell=0}^{\MM-1}\;\sum_{\substack{\omega\in\Omega^{(\MM)} \\ \hat\ell^{\ast}({\bm\theta},\omega)=\ell} }\,
                 \P\big(\{\omega\}\big) \cdot \sum\limits^{\MM}_{i=\ell + 1}\,
                  \log\!\left(1\,+\,s\,\langle\,{\bm t}_{\omega_i\emptyvar}^\top,\,{\bm\theta}\,\rangle\right)
\\ [1.0cm]
  & \overset{\text{\tiny\eqref{eq:ToppingPoint}}}{=}
      &\displaystyle \sum_{\ell=0}^{\MM-1}\hspace{-0.6cm} \sum_{\substack{\omega\in\Omega^{(\MM)} \\
                     \sum\limits_{j=k}^\ell \langle\,{\bm t}_{\omega_j\emptyvar}^\top,\, {\bm\theta}\,\rangle\,>\,0\ \text{for\,$k=1,\ldots,\ell$} \\
                    \sum\limits_{j=\ell+1}^{\tilde{k}} \langle\,{\bm t}_{\omega_j\emptyvar}^\top,\, {\bm\theta}\,\rangle\,\leq\,0\
                     \text{for\,$\tilde{k}=\ell+1,\ldots,\MM$}}}                                                                             \hspace{-1.6cm}
                    \P\big(\{\omega\}\big) \cdot \sum_{i=\ell+1}^\MM\log\!\rk{1 + s\,\langle\,{\bm t}_{\omega_i\emptyvar}^\top,\,{\bm\theta}\,\rangle}
\\ [2.6cm]
  & \overset{\text{\tiny\eqref{bewzeile}}}{=}
      & d^{(\MM)}_{\textnormal{cur}}(s,{\bm\theta})\,.
 \end{array}
 \end{align*}
 The second inequality in \eqref{eq:5.26-observation} follows as in Section~\ref{sec:3} from $\log(1 + x) \leq x$ (see \eqref{eq:EDD} and \eqref{eq:log-series})
 and the third inequality is already clear from Remark~\ref{rem:drawdownlog-series}.

\end{proof}


\vspace*{0.1cm}
            \section{Conclusion}  \label{sec:6}   

 Let us summarize the results of the last Sections. We obtained two down--trade log series related admissible convex risk measures (ACRM)
 according to Definition~\ref{def1:admissConvex}, namely                                                                             \vspace*{-0.3cm}
 \begin{align*}
    \rho_{\textnormal{down}}({\vek})\,\geq\,\rho_{\textnormal{down}X}({\vek})\,\geq\,0 \quad\text{for all}\quad {\vek} \in \overset{\circ}{\msupp}\,,
 \end{align*}
 see Corollary~\ref{coro:theorem3.8} and Theorems~\ref{theo1:ACRM} and \ref{theo2:ACRM}. Similarly we obtained two current drawdown related (ACRM),
 namely                                                                                                                              \vspace*{-0.3cm}
 \begin{align*}
    \rho_{\textnormal{cur}}({\vek})\,\geq\,\rho_{\textnormal{cur}X}({\vek})\,\geq\,0 \quad\text{for all}\quad {\vek} \in \overset{\circ}{\msupp}\,,
 \end{align*}
 cf. Theorems~\ref{theo4:ACRM} and \ref{theo5:ACRM} as well as Theorem~\ref{theo:5.11-observation}.
 Furthermore, due to Remark~\ref{rem:drawdownlog-series} we have the ordering
 \begin{align}\label{eq:curX-downX} 
    \rho_{\textnormal{cur}}({\vek})\,\geq\,\rho_{\textnormal{down}}({\vek}) \quad\text{and}\quad
     \rho_{\textnormal{cur}X}({\vek})\,\geq\,\rho_{\textnormal{down}X}({\vek})\,,\; {\vek} \in \overset{\circ}{\msupp}\,.
 \end{align}
\noindent

 All four risk measures can be used in order to apply the general framework for portfolio theory of \cite{maier:gfpt2017}.
 Since the  two approximated risk measures $\rho_{\textnormal{down}X}$ and $\rho_{\textnormal{cur}X}$ are positive homogeneous, according to \cite{maier:gfpt2017},
 the efficient portfolios will have an affine linear structure. Although we were able to prove a lot of results for these for practical applications
 relevant risk measures, there are still open questions. To state only one of them, we note that convergence of these risk measures for $K \to \infty$ is
 unclear, but empirical evidence seems to support such a statement (see Figure~\ref{figConvergence}).

  \begin{figure}[htb]
    \centering
    \begin{minipage}[c]{0.95\linewidth}
        \centering
\includegraphics[width=0.8\textwidth]{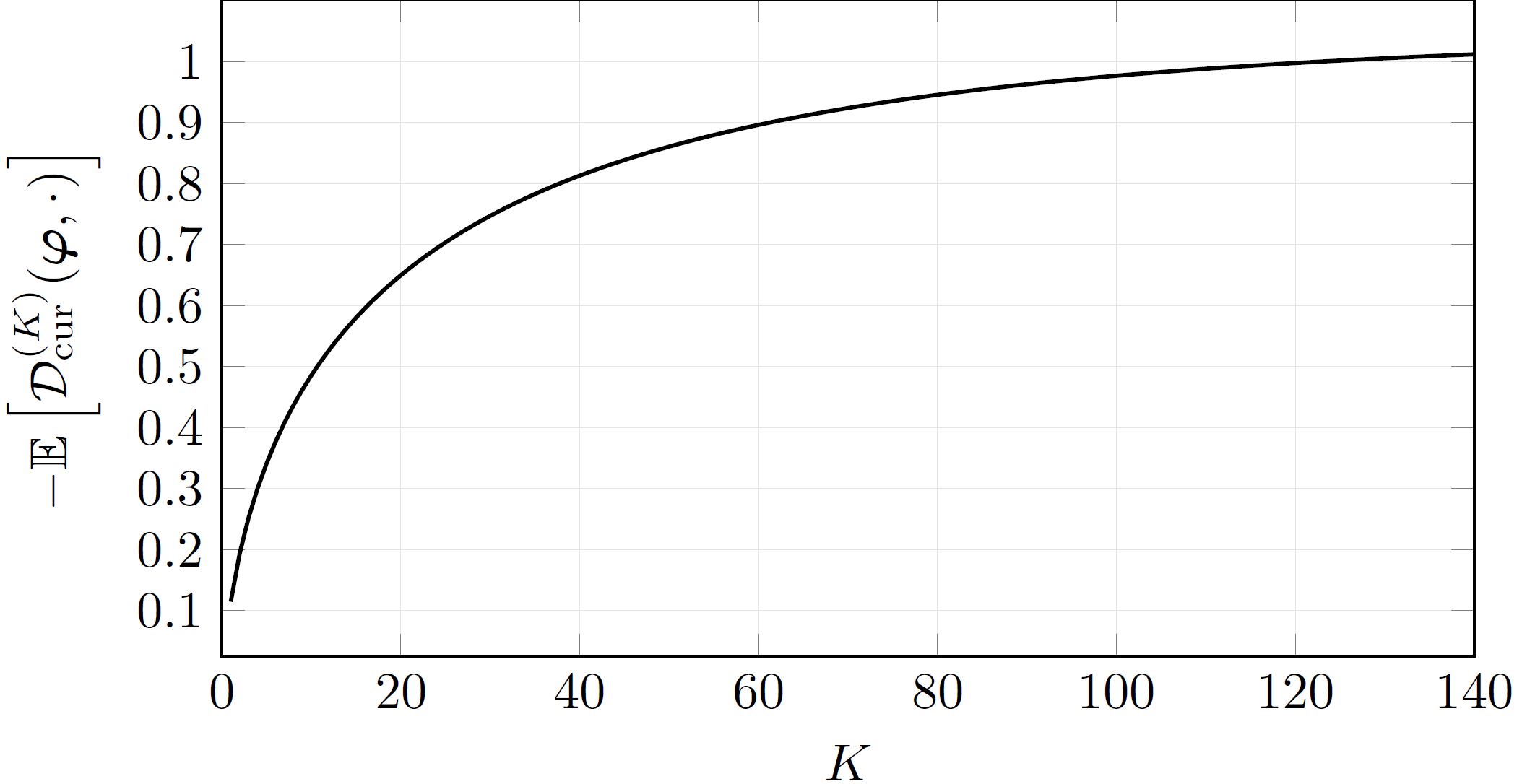}
    \end{minipage}
    \caption{Convergence of $\rho^{(\MM)}_{\textnormal{cur}}$ with fixed
             ${\vek}^{\ast}=({\vek}^{\ast}_1,\,{\vek}^{\ast}_2) ^T= \left(\,\frac{1}{5}\,,\,\frac{1}{5}\right)^T$
             for Example~\ref{exa:risk measures} }
    \label{figConvergence}
  \end{figure}

\begin{appendix}

\vspace*{0.5cm}

\section{\large\bfseries Transfer of a one--period financial market to the TWR setup} \label{secA1}

\vspace*{0.1cm}

 The aim of this appendix is to show that a one--period financial market can be transformed into the Terminal Wealth Relative (TWR)
 setting of Ralph Vince \cite{vince:mmm92} and \cite{vince:rpm09}. In particular we show how the trade return matrix $T$ of
 \eqref{eq:returnMatrix} has to be defined in order to apply the risk measure theory for current drawdowns of Section \ref{sec:4}
 and \ref{sec:5} to the general framework for portfolio theory of Maier-Paape and Zhu of Part I \cite{maier:gfpt2017}.

\vspace*{0.3cm}

\begin{setup}\label{setupA.1} 
  \textbf{(one--period financial market)} \\
  Let\ $S_t = \left(S^0_t,S^1_t,\ldots,\,S^M_t\right)\,,\ t \in \{0,1\}$\ be a financial market in a one--period economy.
 Here\ $S^0_0 = 1$ and $S^0_1 = R \geq 1$\ represents a risk free bond, whereas the other components\ $S^m_t,m = 1,\ldots,\,M$\
  represent the price of the $m$--th risky asset at time $t$ and\ $\Shat_t = \left(S^1_t,\ldots,\,S^M_t\right)$\ is the vector
  of all risky assets. $S_0$ is assumed to be a constant vector whose components are the prices of the assets at\ $t = 0$.
 Furthermore\ $\Shat_1 = \left(S^1_1,\ldots,\,S^M_1\right)$\ is assumed to be a random vector on a finite probability space\
  $\cA = \cA_N = \left\{\alpha_1,\ldots,\,\alpha_N\right\}$, i.e.\ $\Shat_1\!:\,\cA_N \to \R^M$\ represents the new price at\ $t = 1$\
  for the risky assets.
\end{setup}
\vspace*{0.1cm}

\begin{assumption}\label{assum:A.2} 
  To avoid redundant risky assets, often the matrix
  \begin{align}\label{period financial} 
     \widehat{T}_S\,=\,\begin{bmatrix} \\
                          S^1_1\!\left(\alpha_1\right) - R\,S^1_0 & & S^2_1\!\left(\alpha_1\right) - R\,S^2_0 & \ldots & S^M_1\!\!\left(\alpha_1\right) - R\,S^M_0 \\ \\
                          S^1_1\!\left(\alpha_2\right) - R\,S^1_0 & & S^2_1\!\left(\alpha_2\right) - R\,S^2_0 & \ldots & S^M_1\!\!\left(\alpha_2\right) - R\,S^M_0 \\ \\
                          \vdots                                  & & \vdots                                  &        & \vdots                                    \\ \\
                          S^1_1\!\left(\alpha_N\right) - R\,S^1_0 & & S^2_1\!\left(\alpha_N\right) - R\,S^2_0 & \ldots & S^M_1\!\!\left(\alpha_N\right) - R\,S^M_0 \\
                       \end{bmatrix} \in \R^{N \times M}
  \end{align}

\noindent
  is assumed to have full rank $\KK$, in particular $N \geq \KK$.
\end{assumption}

\noindent
  A portfolio is a column vector\ $x \in \R^{M+1}$\ whose components $x_m$ represent the investments in the $m$--th asset,\
  $m=0,\ldots,\,M$. In order to normalize that situation, we consider portfolios with unit initial cost, i.e.
  \begin{align}\label{eq:portfolio} 
     S_0 \cdot x = 1\,.
  \end{align}

\vspace*{0.1cm}\noindent
  Since\ $S^0_0 = 1$\ this implies \vspace*{-0.2cm}
  \begin{align}\label{eq:initial-cost} 
     x_0 + \Shat_0 \cdot \xhat = x_0 + \sum\limits^M_{m=1}\,S^m_0 x_m = 1\,.
  \end{align}

\noindent
 Therefore the interpretation in Table~\ref{tab:capital fraction} is obvious.

 \begin{table}[h]
 \begin{displaymath}\hspace{-0.1cm}
 \begin{tabular}{|@{\hspace{0.2cm}}c@{\hspace{0.2cm}}|l@{\hspace{0.2cm}}| }                              \hline
                   &                                                                          \\ [-0.2cm]
     $x_0$         &  portion of capital invested in bond                                    \\ [0.30cm]
     $S^m_0\,x_m$  &  portion of capital invested in $m$--th risky asset, $m=1,\ldots,\,M$   \\ [-0.2cm]
                   &                                                                          \\         \hline
 \end{tabular}
 \end{displaymath}
   \centering
   \caption{Invested capital portions}
   \label{tab:capital fraction}
 \end{table}

\noindent
 So if an investor has an initial capital of $C_{ini}$ in his depot, the invested money in the depot is divided
 as in Table~\ref{tab:moneydepot}.


 \begin{table}[h]
 \begin{displaymath}\hspace{-0.1cm}
 \begin{tabular}{|@{\hspace{0.25cm}}l@{\hspace{0.2cm}}|l@{\hspace{0.2cm}}| }                                     \hline
                           &                                                                               \\ [-0.2cm]
     $C_{ini}\,x_0$         &  cash position of the depot                                                   \\ [0.30cm]
     $C_{ini}\,S^m_0\,x_m$  &  invested money in $m$--th asset,\ $m=1,\ldots,\,M$                           \\ [0.30cm]
     $C_{ini}\,x_m$         &  amount of shares of $m$--th asset to be bought at\ $t=0\,,\ m=1,\ldots,\,M$  \\ [-0.2cm]
                           &                                                                               \\         \hline
 \end{tabular}
 \end{displaymath}
   \centering
   \caption{Invested money in depot for a portfolio $x$}
   \label{tab:moneydepot}
 \end{table}

\vspace*{0.3cm}
 Clearly\ $\left(S_1 - R\,S_0\right) \cdot x = S_1 \cdot x - R$\ is the (random) gain of the unit initial cost portfolio
 relative to the riskless bond. In such a situation the merit of a portfolio $x$ is often measured by its
 expected utility\ $\E\big[u\left(S_1 \cdot x\right)\big]$,\ where $u$ is an increasing concave utility function
 (see \cite{maier:gfpt2017}, Assumption~3.3).
 In growth optimal portfolio theory the natural logarithm\ $u = \log$\ is used yielding the optimization problem
\newpage
 \begin{align}\label{eq:natural-logarithm} 
 \begin{array}{c@{\hspace{0.4cm}}r@{\hspace{0.18cm}}c@{\hspace{0.18cm}}l}
                & \E\ek{\log\!\left(S_1 \cdot x\right)} & \overset{\text{\large!}}{=} &  \max\,,\qquad x \in \R^{M+1}, \\ [0.3cm]
    \text{s.t.} & S_0 \cdot x                                            &      =                      & 1\,.
 \end{array}
 \end{align}

 The following discussion aims to show that the above optimization problem \eqref{eq:natural-logarithm} is an alternative way
 of stating the Terminal Wealth Relative optimization problem of Vince (cf. \cite{hermes:twr2017}, \cite{vince:nmm1995}). 

\vspace*{0.2cm}\noindent
 Using\ $S^0_1 = R$\ we obtain\ $S_1 \cdot x = R\,x_0 + \Shat_1 \cdot \xhat$\ and hence with \eqref{eq:initial-cost}
 \begin{align*}
 \begin{array}{l@{\hspace{0.18cm}}c@{\hspace{0.18cm}}l}
    \E\Big[\log\!\left(S_1 \cdot x\right)\Big]
         & = &\displaystyle \E\Big[\log\!\left(R\big(1 - \Shat_0\cdot\xhat\big) + \Shat_1\cdot\xhat\right)\Big] \\ [0.5cm]
         & = &\displaystyle \sum\limits_{\alpha\,\in\,\cA_N}\,\Pe\Big(\{\alpha\}\Big) \cdot
                          \log\Big(R + \big[\Shat_1(\alpha) - R\,\Shat_0\big] \cdot \xhat\Big)\,.
 \end{array}
 \end{align*}

\noindent
 Assuming all\ $\alpha \in \cA_N$\ have the same probability (Laplace situation), i.e.
 \begin{align}\label{eq:Laplace-situation} 
   \Pe\Big(\left\{\alpha_i\right\}\Big) = \frac{1}{N} \quad\text{for all}\quad i = 1,\ldots,N\,,
 \end{align} 

 we furthermore get
 \begin{align}\label{eq:Laplace-Sum^N} 
 \begin{array}{l@{\hspace{0.18cm}}c@{\hspace{0.18cm}}l}
    \E\ek{\log\!\left(S_1 \cdot x\right)} - \log(R)
         & = &\displaystyle \frac{1}{N}\,\sum\limits^N_{i=1}\
                  \log\!\left(1 + \left[\,\frac{\Shat_1\!\left(\alpha_i\right) - R\Shat_0}{R}\right]\cdot\xhat\right)          \\ [0.5cm]
         & = &\displaystyle \frac{1}{N}\,\sum\limits^N_{i=1}\ \log\!\left(1 + \sum\limits^M_{m=1}\
                  \underbrace{\left[\frac{S^m_1\!\left(\alpha_i\right) - R\,S^m_0}{R\,S^m_0}\right]}_{=:\,t_{i,m}} \cdot
                  \underbrace{S^m_0\,x_m}_{=:\,\varphi_m}\right)\,.
 \end{array}
 \end{align}

\vspace*{0.2cm}\noindent
 This results in a ``trade return" matrix
 \begin{align}\label{eq:trade-return} 
    T = \left(t_{i,m}\right)_{\substack{1 \leq i \leq N \\ 1 \leq m \leq M}} \in \R^{N \times M}
 \end{align}
 whose entries represent discounted relative returns of the $m$--th asset for the $i$--th reali-\\zation $\alpha_i$. Furthermore,
 the column vector\ ${\bm\varphi} = \left(\varphi_m\right)_{1 \leq m \leq M} \in \R^M$\ with components\ $\varphi_m = S^m_0\,x_m$\
 has according to Table~\ref{tab:capital fraction} the interpretation given in Table~\ref{tab:TWRmodel}.


 \begin{table}[h]
 \begin{displaymath}\hspace{-0.1cm}
 \begin{tabular}{|@{\hspace{0.2cm}}c@{\hspace{0.2cm}}|l@{\hspace{0.2cm}}| }                                 \hline
                        &                                                                        \\ [-0.1cm]
     $\varphi_m$        & portion of capital invested in $m$--th risky asset,\ $m=1,\ldots,M$   \\ [-0.1cm]
                        &                                                                        \\         \hline
 \end{tabular}
 \end{displaymath}
   \centering
   \caption{Investment vector ${\bm\varphi}$ for the $\TWR$ model}
   \label{tab:TWRmodel}
 \end{table}


\vspace*{0.2cm}\noindent
 Thus we get
 \begin{align}\label{eq:column} 
 \begin{array}{l@{\hspace{0.18cm}}c@{\hspace{0.18cm}}l}
    \E\big[\log\!\left(S_1 \cdot x\right)\big] - \log(R)
         & = &\displaystyle \frac{1}{N}\,\sum\limits^N_{i=1}\
                     \log\!\left(1\,+\,\langle\,{\mathbf t}_{i\emptyvar}^\top,{\bm\varphi}\rangle_{\R^M}\right)      \\ [0.5cm]
         & = &\displaystyle \log\left(\left[\prod^N_{i=1}\,\left(1\,+\,\langle\,{\mathbf t}_{i\emptyvar}^\top,{\vek}\rangle_{\R^M}\right)
                     \right]^{\!\!1/N}\right)=   \log\!\left( \left[\TWR^{(N)}(\bm\varphi)  \right]^{\!\!1/N}\right)
 \end{array}
 \end{align}

\noindent 
 which involves the usual Terminal Wealth Relative ($\TWR$) of Ralph Vince \cite{vince:nmm1995} and therefore under the assumption of a
 Laplace situation \eqref{eq:Laplace-situation} the optimization problem \eqref{eq:natural-logarithm} is equivalent to
 \begin{align}\label{eq:TWR-Vince} 
    \TWR^{(N)}(\vek)\,\overset{\textrm{\large!}}{=} \,\max\,,\quad {\bm\varphi} \in \R^M\,.
 \end{align} 

\vspace*{0.1cm}
 Furthermore, the trade return matrix $T$ in \eqref{eq:trade-return} may be used to define admissible convex risk measures as
 introduced in Definition~\ref{def1:admissConvex} which in turn give nontrivial applications to the general framework for
 portfolio theory in Part I \cite{maier:gfpt2017}.

\vspace*{0.3cm}
 To see that, note again that by \eqref{eq:Laplace-Sum^N} any portfolio vector\ $x=\left(x_0,\widehat{x}\right)^T \in \R^{\KK}$\
 of a unit cost portfolio \eqref{eq:portfolio} is in one to one correspondence to an investment vector
 \begin{align}\label{eq:investmentVector} 
    {\vek}\,=\,\left({\vek}_m\right)_{1\,\leq\,m\,\leq\,\KK}\,=\,\left(S_0^m \cdot x_m\right)_{1\,\leq\,m\,\leq\,\KK}\,=:\,\Lambda \cdot \widehat{x}
 \end{align}
 for a diagonal matrix\ $\Lambda \in \R^{\KK\times\KK}$ with only positive diagonal entries $\Lambda_{m,m} = S^m_0$.
\noindent
Then we obtain: \vspace*{-0.18cm}
\begin{theorem}\label{theo:diagonalMatrix} 
   Let\ $\rho: \textnormal{Def}(\rho) \to \R^+_0$\ be any of our four down--trade or drawdown related risk measures\
   $\rho_{\textnormal{down}},\,\rho_{\textnormal{down}X},\,\rho_{\textnormal{cur}}$ and\ $\rho_{\textnormal{cur}X}$ (see \eqref{eq:curX-downX})
   for the trading game of Setup~\ref{setup:Z} satisfying Assumption~\ref{no-risk-free-investment}. 
   Then 
   \begin{align}\label{eq:A.11-Lambda} 
      \widehat{\rho}\left(\widehat{x}\right) := \rho\left(\Lambda \widehat{x}\right) = \rho({\vek})\,,\quad
      \widehat{x} \in \textnormal{Def}\left(\widehat{\rho}\right) := \Lambda^{-1}\,\textnormal{Def}(\rho) \subset \R^{\KK}
   \end{align}
   has the following properties: \vspace*{-0.18cm}
   \begin{enumerate} \leftskip 0.35cm
   \item[\textit{(r1)}]
         $\widehat{\rho}$ depends only on the risky part $\widehat{x}$ of the portfolio\ $x = \left(x_0,\widehat{x}\right)^T\in \R^{\KK+1}$.
   \item[\textit{(r1n)}]
         $\widehat{\rho}\left(\widehat{x}\right) = 0$\ if and only if\ $\widehat{x} = \widehat{0} \in \R^{\KK}$.
   \item[\textit{(r2)}]
         $\widehat{\rho}$ is convex in $\widehat{x}$.
   \item[\textit{(r3)}]
         The two approximations $\rho_{\textnormal{down}X}$ and $\rho_{\textnormal{cur}X}$ furthermore yield \\
         positive homogeneous $\widehat{\rho}$, i.e. $\widehat{\rho}\left(t\widehat{x}\right) = t\widehat{\rho}\left(\widehat{x}\right)$
         for all $t > 0$.
   \end{enumerate}
\end{theorem}
\begin{proof}
  See the respective properties of $\rho$ (cf. Theorems~\ref{theo1:ACRM}, \ref{theo2:ACRM}, \ref{theo4:ACRM} and \ref{theo5:ACRM}).
  In particular\ $\rho_{\textnormal{down}},\,\rho_{\textnormal{down}X},\,\rho_{\textnormal{cur}}$ and\ $\rho_{\textnormal{cur}X}$\
  are admissible convex risk measures according to Definition~\ref{def1:admissConvex} and thus\
  $\textit{(r1)}, \textit{(r1n)}, \text{and} \textit{(r2)}$ follow. 
\end{proof}
\begin{remark}\label{rem:theoA.3} 
  It is clear that therefore\
  $\widehat{\rho}=\widehat{\rho}_{\textnormal{down}},\,\widehat{\rho}_{\textnormal{down}X},\,\widehat{\rho}_{\textnormal{cur}}$ or\ $\widehat{\rho}_{\textnormal{cur}X}$
  can be evaluated on any set of admissible portfolios $A \subset \R^{\KK + 1}$ according to Definition~2.2 of \cite{maier:gfpt2017} if
  \begin{align*}
     \text{Proj}_{\R^{\MM}}\,A\,\subset\,\textnormal{Def}(\widehat{\rho})
  \end{align*}
  and the properties \textit{(r1)}, \textit{(r1n)}, \textit{(r2)} (and only for $\rho_{\textnormal{down}X}$ and $\rho_{\textnormal{cur}X}$ also \textit{(r3)})
  in \\ Assumption~3.1 of \cite{maier:gfpt2017} follow from Theorem~\ref{theo:diagonalMatrix}.
In particular\ $\widehat{\rho}_{\textnormal{down}X}$\ and\ $\widehat{\rho}_{\textnormal{cur}X}$\ satisfy the conditions of a deviation
  measure in \cite{rockafellar:mfpa2006} (which is defined directly on the portfolio space). 
\end{remark}

\begin{remark}\label{rem:drawdown-function} 
  Formally our drawdown or down--trade is a function of a TWR equity curve of a $K$\!--period financial market.
  But since this equity curve is obtained by drawing $K$ times stochastically independent from one and the
  same market in Setup~\ref{setupA.1}, we still can work with a one--period market model.
\end{remark}
 We want to close this section with some remarks on the often used no arbitrage condition of the one--period financial  market
 and Assumption \ref{no-risk-free-investment} which was necessary to construct admissible convex risk measures.

\begin{definition}\label{def:one-period-financial-market} 
  Let $S_t$ be a one--period financial market as in Setup~\ref{setupA.1}.
  \begin{enumerate}
  \item  A portfolio\ $x \in \R^{\KK+1}$\ is an \textbf{arbitrage} for $S_t$ if it satisfies
         \begin{align} \label{eq1a:one-period-financical-market} 
            \Big(S_1 - RS_0\Big) \cdot x \geq 0 \quad\text{and}\quad \Big(S_1 - RS_0\Big) \cdot x \not= 0\,,
         \end{align}
         or, equivalently, if\; $\xhat \in \R^{\KK}$\ is satisfying
         \begin{align}\label{eq1b:one-period-fincial} 
            \left(\Shat_1 - R\,\Shat_0\right) \cdot \xhat\,\geq\,0 \quad\text{and}\quad
            \left(\Shat_1 - R\,\Shat_0\right) \cdot \xhat\,\not=\,0\,.
         \end{align}
  \item  The market $S_t$ is said to have \textbf{no arbitrage}, if there exists no arbitrage portfolio.
  \end{enumerate}
\end{definition}
\noindent
 Once we consider the above random variables as vector\
 $\left[\,\left(\Shat_1\!\left(\alpha_i\right) - R\,\Shat_0\right) \cdot \xhat\right]_{1\,\leq\,i\,\leq\,N} \in \R^N$,\
 \eqref{eq1b:one-period-fincial} may equivalently be stated as
  \begin{align}\label{eq:random-variable} 
            \left(\Shat_1 - R\,\Shat_0\right) \cdot \xhat \in \cK \setminus \{0\}\,,
 \end{align}

\noindent
 where we used the positive cone\ $\cK := \left\{y \in \R^N\!: y_i \geq 0\;\text{for}\;i=1,\ldots,N\right\}$\ in $\R^N$.  \\
 Observe that the portfolio $\xhat = 0$ can never be an arbitrage portfolio.

\vspace*{0.1cm}\noindent
 Hence we get: \\ [0.1cm]\hspace*{0.5cm}
  market\ $S_t$\ has no arbitrage
  \begin{align}\hspace*{-4.5cm}
    & \Longleftrightarrow \quad \forall\ \xhat \in \R^M \setminus \{0\} \quad \text{holds} \quad
                         \left(\Shat_1 - R\,\Shat_0\right) \cdot \xhat \notin \cK \setminus \{0\}
      \label{eq12:market-arbitrage} 
    \\
    & \Longleftrightarrow \quad \forall\ \xhat \in \R^M \setminus \{0\} \quad \text{holds} \quad
                         \left(\Shat_1 - R\,\Shat_0\right) \cdot \xhat \notin \big(\cK\,\cup\,(-\cK)\big) \setminus \{0\}\,,
      \label{eq14:market-arbitrage}\hspace*{0.2cm} 
 \end{align}
\hspace*{0.5cm}
 where we used the negative cone\ $(-\cK) = \left\{y \in \R^N\!: y_i \le 0\;\text{for}\;i=1,\ldots,N\right\}$ in $\R^N$.

\vspace*{0.2cm}\noindent
 Note that the last equivalence leading to \eqref{eq14:market-arbitrage} follows, because with\ $\xhat \in \R^M \setminus \{0\}$\
 always\ $\left(-\xhat\right) \in \R^M \setminus \{0\}$\ also holds true. 
 According to Setup~\ref{setupA.1}, the matrix $\widehat{T}_S$ has full rank, and therefore\
 $\left(\Shat_1 - R\,\Shat_0\right) \cdot \xhat \not= 0$\ for all\ $\xhat \not= 0$\ anyway.
 Hence we proceed
 \begin{align}\label{eq:setupA1} 
 \begin{array}{@{\hspace{0.1cm}}c@{\hspace{0.2cm}}c@{\hspace{0.2cm}}l@{\hspace{0.2cm}}l@{\hspace{0.2cm}}l}
   \eqref{eq14:market-arbitrage}
     & \Longleftrightarrow & \forall\ \xhat \in \R^M \setminus \{0\} & \text{holds}\
                                       \left(\Shat_1 - R\,\Shat_0\right) \cdot \xhat \notin \big(\cK\,\cup\,(-\cK)\big)    \\ [0.3cm]
     & \Longleftrightarrow & \forall\ \xhat \in \R^M \setminus \{0\} & \text{exists some\; $\alpha_{i_0} \in \cA_N$\; with}
                                      \left(\Shat_1\!\left(\alpha_{i_0}\right) - R\,\Shat_0\right) \cdot \xhat < 0        \\ [0.3cm]
     &                     &                                         & \text{and some\; $\alpha_{j_0} \in \cA_N$\; with}
                                      \left(\Shat_1\!\left(\alpha_{j_0}\right) - R\,\Shat_0\right) \cdot \xhat > 0\,.
 \end{array}
 \end{align}

\noindent
 Observe that for all\ $\xhat \in \R^M \setminus \{0\}$\ and\ $i_0 \in \{1,\ldots,N\}$\ the following is equivalent 
 \begin{align}\label{eq6:equitable} 
 \begin{array}{l@{\hspace{0.3cm}}c@{\hspace{0.3cm}}l}
    \left(\Shat_1\!\left(\alpha_{i_0}\right) - R\,\Shat_0\right) \cdot \xhat < 0
       & \Longleftrightarrow &\displaystyle \sum\limits^M_{m=1}\ \underbrace{\left[\frac{S^m_1\!\left(\alpha_{i_0}\right) -
                                 R\,S^m_0}{R\,S^m_0}\right]}_{=\,t_{i_0,m}} \cdot \underbrace{S^m_0\,x_m}_{=\,\vek_m} < 0  \\ [0.3cm]
       & \overset{\eqref{eq:Laplace-Sum^N}}{\Longleftrightarrow}
                             &\displaystyle \langle\,{\mathbf t}_{i_0\emptyvar}^\top,{\bm\varphi}\rangle < 0\,.
 \end{array}
 \end{align}

\noindent
 Hence
 \begin{align}\label{eq:setupB1} 
 \begin{array}{@{\hspace{-0.2cm}}c@{\hspace{0.25cm}}c@{\hspace{0.25cm}}l@{\hspace{0.2cm}}l@{\hspace{0.2cm}}l@{\hspace{0.2cm}}l}
   \eqref{eq:setupA1}
     & \Longleftrightarrow  & \forall\ {\bm\varphi} \in \R^M \setminus \{0\} & \text{exists some}\;\; i_0,j_0 \in \{1,\ldots,N\}
                                                          & \text{with} & \langle\,\mathbf{t}_{i_0\emptyvar}^\top,{\bm\varphi}\rangle_{\R^M} < 0  \\ [0.3cm]
     &                      &      &                      & \text{and}  & \langle\,\mathbf{t}_{j_0\emptyvar}^\top,{\bm\varphi}\rangle_{\R^M} > 0
 \end{array}
 \end{align}
 \vspace*{-0.1cm}
 \begin{align}\label{eq:setupC1} 
 \begin{array}{@{\hspace{0.9cm}}c@{\hspace{0.2cm}}l@{\hspace{0.2cm}}l@{\hspace{0.2cm}}l@{\hspace{0.2cm}}l}
      \Longleftrightarrow & \forall\ {\bm\theta} \in {\Se}^{M-1}_1 & \text{exists some}\;\; i_0 \in \{1,\ldots,N\}
            & \text{with} & \langle\,\mathbf{t}_{i_{0}\emptyvar}^\top,{\bm\theta}\rangle_{\R^M} < 0
 \end{array}
 \end{align}
 using again the argument that\ ${\bm\theta} \in {\Se}^{M-1}_1$ also implies\ $\left(-{\bm\theta}\right) \in {\Se}^{M-1}_1$.
 To conclude, \eqref{eq:setupC1} is exactly \textit{Assumption~\ref{no-risk-free-investment}} and therefore we get:
\begin{theorem}\label{theo:financial-market} 
  Let a one--period financial market $S_t$ as in Setup~\ref{setupA.1} be given that satisfies  Assumption~\ref{assum:A.2}, i.e. $\widehat{T}_S$ from
\eqref{period financial}   has full rank $\KK$. Then the market $S_t$ has no arbitrage if and only if the in \eqref{eq:Laplace-Sum^N} and \eqref{eq:trade-return}
  derived trade return matrix $T$ satisfies Assumption~\ref{no-risk-free-investment}.
\end{theorem}

\noindent
 A very similar theorem is derived in \cite{maier:gfpt2017}.
 For completeness we rephrase here that part which is important in the following.
\begin{theorem}(\cite{maier:gfpt2017}, Theorem~3.9)\label{theo:one-period-fin-market} 
  Let a one--period financial market $S_t$ as in Setup~\ref{setupA.1} with no arbitrage be given. Then the conditions in Assumption~2.3
  are satisfied for $T$ from \eqref{eq:Laplace-Sum^N} and \eqref{eq:trade-return} if and only if Assumption~\ref{assum:A.2} holds, i.e.
  if $\widehat{T}_S$ from \eqref{period financial}  has full rank $\KK$.
\end{theorem} 
\begin{proof}
  Note that the conditions in Assumption~\ref{no-risk-free-investment} for $T$ from \eqref{eq:Laplace-Sum^N} and \eqref{eq:trade-return}
  are equivalent to
  \begin{align}\label{eq:proof-theoA.8} 
     \begin{array}{l}
      \text{``for every risky portfolio\ $\widehat{x} \not= \widehat{0}$,\ there exists some\ $\alpha \in \cA_N$} \\ [0.1cm]
      \text{such that}\ \left(\Shat_1\!(\alpha) - R\,\Shat_0\right) \cdot \widehat{x} < 0"\,,
  \end{array}
  \end{align}
  which follows directly from the equivalence of \eqref{eq:setupC1} and \eqref{eq:setupA1}. But \eqref{eq:proof-theoA.8} is exactly the
  point (ii*) in \cite{maier:gfpt2017}, Theorem~3.9, and Assumption~\ref{assum:A.2} is exactly the point (iii) of that theorem.
  Therefore, under the no arbitrage assumption again by \cite{brenner:pmpt17}, Theorem~3.9, the claimed equivalence follows.   
\end{proof}

\noindent
 Together with Theorem~\ref{theo:financial-market} we immediately conclude: \vspace*{-0.18cm}
\begin{corollary}\label{col:three-imply-the-third} 
  \textbf{(two out of three imply the third)} \\ [0.1cm]
  Let a one--period financial market $S_t$ according to Setup~\ref{setupA.1} be given. Then any two of the following conditions
  imply the third: \vspace*{-0.18cm}
  \begin{enumerate}
   \item Market $S_t$ has no arbitrage.
   \item The trade return matrix $T$ from \eqref{eq:Laplace-Sum^N} and \eqref{eq:trade-return} satisfies Assumption~\ref{no-risk-free-investment}.
   \item Assumption~\ref{assum:A.2} holds, i.e. $\widehat{T}_S$ from \eqref{period financial} has full rank $\KK$.
  \end{enumerate} 
\end{corollary}
\vspace*{-0.5cm} 
\begin{remark}\label{rem:standard-assumption} 
  The standard assumption on the market $S_t$ in Part I \cite{maier:gfpt2017} is ``no nontrivial riskless portfolio",
  where a portfolio\ $x = \left(x_0,\widehat{x}\right)^T \in \R^{\KK+1}$ is \textnormal{riskless} if
  \begin{align*}
     \left(S_1 - R\,S_0\right) \cdot x \geq 0
  \end{align*}
  and $x$ is \textnormal{nontrivial} if\ $\widehat{x} \not= \widehat{0}$.
\end{remark} 
\noindent
 Using this notation we get:\vspace*{-0.2cm}
\begin{corollary}\label{col:one-period-financial-market} 
  Consider a one--period financial market $S_t$ as in Setup~\ref{setupA.1}. Then there is no nontrivial riskless portfolio in $S_t$ if and 
  only if any two of the three statements \textit{(a), (b)}, and \textit{(c)} from Corollary~\ref{col:three-imply-the-third} are satisfied.
\end{corollary}
\begin{proof}
  Just apply \cite{maier:gfpt2017}, Proposition~3.7 together with \cite{maier:gfpt2017}, Theorem~3.9 to the situation of
  Corollary~\ref{col:three-imply-the-third}.
\end{proof} 

\noindent
 To conclude, any two of the three conditions of Corollary~\ref{col:three-imply-the-third} on the market $S_t$ are
 sufficient to apply the theory presented in Part I \cite{maier:gfpt2017}.

\end{appendix}


\vspace*{0.8cm}


\end{document}